\documentclass[11pt]{article}
\usepackage[margin=1in]{geometry}
\usepackage{booktabs} % For formal tables

\usepackage[T1]{fontenc}
\usepackage{authblk}
\usepackage{cite}
\usepackage{amsmath, amssymb, amsthm, thmtools, amsfonts, bm, bbm, thm-restate}
\usepackage[ruled,lined,boxed]{algorithm2e}
\usepackage{enumitem}
\usepackage[normalem]{ulem}

\usepackage{asymptote}
\usepackage{graphicx}
\usepackage{todonotes}
\usepackage{multirow}
\usepackage{comment}
\usepackage{dsfont}
\usepackage{bigstrut}
\usepackage{caption}
\usepackage{subcaption}

\usepackage{xcolor}
\definecolor{ForestGreen}{rgb}{0.1333,0.5451,0.1333}
\definecolor{DarkRed}{rgb}{0.65,0,0}
\definecolor{Red}{rgb}{1,0,0}
\PassOptionsToPackage{hyphens}{url}
\usepackage[linktocpage=true,backref=page,
pagebackref=true,colorlinks,
linkcolor=DarkRed,citecolor=ForestGreen,
bookmarks,bookmarksopen,bookmarksnumbered]
{hyperref}

\usepackage{framed}
\usepackage[framemethod=tikz]{mdframed}
\usepackage[skins]{tcolorbox}

\usepackage{cleveref}
\allowdisplaybreaks

\theoremstyle{plain}
\newtheorem{thm}{Theorem}[section]
\newtheorem{cor}[thm]{Corollary}
\newtheorem{prop}[thm]{Proposition}
\newtheorem{conj}[thm]{Conjecture}
\newtheorem{fact}[thm]{Fact}
\newtheorem{lem}[thm]{Lemma}
\newtheorem{lemma}[thm]{Lemma}

\newtheorem{claim}[thm]{Claim}
\newtheorem{Def}[thm]{Definition}
\newtheorem{obs}[thm]{Observation}

\newtheorem{rem}[thm]{Remark}

\crefname{thm}{Theorem}{theorems}
\crefname{cla}{Claim}{claims}
\crefname{lem}{Lemma}{lemmas}
\crefname{fact}{Fact}{facts}

\newcommand{\Ber}{\textrm{Ber}}
\newcommand{\Uni}{\textrm{Uni}}
\newcommand{\Geo}{\textrm{Geo}}

\newcommand{\E}{\mathbb{E}}
\newcommand{\R}{\mathbb{R}}

\newcommand{\Z}{\mathbb{Z}}
\newcommand{\eps}{\varepsilon}
\newcommand{\calA}{\mathcal{A}}

\newcommand{\calC}{\mathcal{C}}

\newcommand{\calF}{\mathcal{F}}

\newcommand{\calR}{\mathcal{R}}
\newcommand{\calS}{\mathcal{S}}
\newcommand{\calW}{\mathcal{W}}

\newcommand{\poly}{\mathrm{poly}}
\newcommand{\suppOp}{\mathrm{supp}}
\newcommand{\support}[1]{\mathrm{supp}({\bf{#1}})}
\newcommand{\supp}[1]{\support{#1}}
\newcommand\norm[1]{\lVert{\bf{#1}}\rVert}
\newcommand\normOp[1]{\lVert{#1}\rVert}
\newcommand{\defeq}{:=}

\def\capitalize#1{\uppercase{#1}}

\newcommand{\adaptive}{adaptive\xspace}
\newcommand{\outputadaptive}{output-adaptive\xspace}

\newenvironment{wrapper}[1]
{
	\smallskip
	\begin{center}
		\begin{minipage}{\linewidth}
			\begin{mdframed}[hidealllines=true, backgroundcolor=gray!20, leftmargin=0cm,innerleftmargin=0.3cm,innerrightmargin=0.3cm,innertopmargin=0.375cm,innerbottommargin=0.375cm,roundcorner=10pt]
				#1}
			{\end{mdframed}
		\end{minipage}
	\end{center}
	\smallskip
}

\newcommand{\codestyle}[1]{\mathtt{#1}}
\newcommand{\dsplit}{\codestyle{degree}\text{-}\codestyle{split}\xspace}

\newcommand{\init}{\codestyle{init}}
\newcommand{\set}{\codestyle{set}}
\newcommand{\sample}{\codestyle{sample}}
\newcommand{\update}{\codestyle{update}}
\newcommand{\rebuild}{\codestyle{rebuild}}
\newcommand{\geom}{\mathrm{geom}}

\newcommand{\x}{{\bf x}}
\newcommand{\y}{{\bf y}}
\newcommand{\z}{{\bf z}}

\newcommand{\DM}{$\eps$-robust\xspace}

\title{Near-Optimal Dynamic Rounding\\ of Fractional Matchings in Bipartite Graphs}

\date{\vspace{-1cm}}

\author[1]{Sayan Bhattacharya\thanks{Supported by Engineering and Physical Sciences Research Council, UK (EPSRC) Grant EP/S03353X/1.}}
\author[1]{Peter Kiss\thanks{Work done in part while the author was visiting Max-Planck-Institut für Informatik.}}
\author[2]{Aaron Sidford\thanks{Supported in part by a Microsoft Research Faculty Fellowship, NSF CAREER Award CCF-1844855, NSF Grant CCF-1955039, a PayPal research award, and a Sloan Research Fellowship.}}
\author[3]{David Wajc\thanks{Work done in part while the author was at Stanford University and Google Research. Supported by a Taub Family Foundation ``Leader in Science and Technology'' fellowship.}}
\affil[1]{University of Warwick}
\affil[2]{Stanford University}
\affil[3]{Technion --- Israel Institute of Technology}

\begin{document}
\maketitle
\pagenumbering{gobble}
\begin{abstract}
We study dynamic $(1-\epsilon)$-approximate  rounding of fractional matchings---a key ingredient in numerous breakthroughs in the dynamic graph algorithms literature. Our first contribution is a surprisingly simple deterministic rounding algorithm in bipartite graphs with amortized update time $O(\epsilon^{-1} \log^2 (\epsilon^{-1} \cdot n))$, matching an (unconditional) recourse lower bound of $\Omega(\epsilon^{-1})$ up to logarithmic factors. Moreover, this algorithm's update time improves provided the minimum (non-zero) weight in the fractional matching is lower bounded throughout. Combining this algorithm with novel dynamic \emph{partial rounding} algorithms to increase this minimum weight, we obtain a number of algorithms that improve this dependence on $n$. For example, we give a high-probability randomized algorithm with $\tilde{O}(\epsilon^{-1}\cdot (\log\log n)^2)$-update time against adaptive adversaries.\footnote{Throughout, we use ``soft-O'' notation, $\tilde{O}$, to suppress logarithmic factors in $\eps$, i.e., $\tilde{O}(f)=O(f\cdot \mathrm{poly}(\log (\eps^{-1})))$.} Using our rounding algorithms, we also round known $(1-\epsilon)$-decremental fractional bipartite matching algorithms with no asymptotic overhead, thus improving on state-of-the-art algorithms for the decremental bipartite matching problem. Further, we provide extensions of our results to general graphs and to maintaining almost-maximal matchings.
\end{abstract}

\newpage
\tableofcontents

\newpage

\pagenumbering{arabic}
\section{Introduction}\label{sec:intro}

Dynamic matching is one of the most central and well-studied dynamic algorithm problems.
Here, a graph undergoes edge insertions and deletions, and we wish to quickly update a large matching (vertex-disjoint set of edges) following each such change to the graph.

A cornerstone of numerous dynamic matching results 
is the dynamic \emph{relax-and-round} approach: the combination of dynamic \emph{fractional} matching algorithms \cite{bhattacharya2015deterministic,bhattacharya2016new,bhattacharya2017fully,bhattacharya2019deterministically,bhattacharya2020deterministic} with dynamic \emph{rounding} algorithms \cite{arar2018dynamic,wajc2020rounding,bhattacharya2021deterministic,kiss2022improving,bhattacharya2023dynamic}.
This dynamic fractional matching problem asks to maintain a vector $\x\in \mathbb{R}^E_{\geq 0}$ such that $x(v)\defeq \sum_{e\ni v} x_e$ satisfies the fractional degree constraint $x(v) \leq 1$ for all vertices $v\in V$ and $\norm{x}:= \sum_e x_e$ is large compared to
the size of the largest  (fractional) matching in the dynamic graph $G=(V,E)$.  The goal typically is to solve this problem while 
 minimizing the amortized or worst-case time per edge update in $G$.\footnote{An algorithm has \emph{amortized} update time $f(n)$ if every sequence of $t$ updates takes at most $t\cdot f(n)$ time and has \emph{worst-case} update time $f(n)$ if each operation takes at most $f(n)$ time. 
As we focus on amortized update times, we omit this distinction.} 
%for brevity, unless we state so explicitly, stated update time bounds are all amortized.} 
For the rounding problem (the focus of this work), an abstract interface can be defined as follows.
\vspace{-0.55cm}
\begin{wrapper}
\begin{Def}\label{def:rounding} A \underline{\emph{dynamic rounding algorithm (for fractional matchings)}} is a data structure supporting the following operations:
\begin{itemize}
    \item $\init(G = (V,E), \, \x \in \R^E_{\geq 0},\, \epsilon \in (0,1))$: initializes the data structure for undirected graph $G$ with vertices $V$ and edges $E$, current fractional matching $\x$ in $G$, and target error $\epsilon$.
    \item $\update(e \in E,\, \nu \in [0,1])$: sets $x_e \gets \nu$ under the promise that the resulting $\x$ is a fractional matching in $G$.\footnote{Invoking $\update(e,0)$ essentially deletes $e$ and subsequently invoking $\update(e,\nu)$ for $\nu > 0$ essentially adds $e$ back. So, $G$ might as well be the complete graph on $V$. However, we find the notation $G = (V, E)$ convenient.}
\end{itemize}
The algorithm must maintain a matching $M$ in the support of $\x$, $\support{x}:=\{e\in E\mid x_e>0\}$, such  that $M$ is a $(1 - \epsilon)$-approximation with respect to $\norm{x}:=\sum_e x_e$, i.e.,
\[
M \subseteq \support{x} \text{ , } M \text{ is a matching } \text{ , and }
|M| \geq (1 - \epsilon)\cdot  \norm{x} \,.
\]
\end{Def}
\end{wrapper}

\vspace{-0.15cm}
The combination of fast fractional algorithms with fast dynamic rounding algorithms plays a key role in state-of-the-art time/approximation trade-offs for the dynamic matching problem against an adaptive adversary \cite{wajc2020rounding,bhattacharya2021deterministic,kiss2022improving}, including the recent breakthroughs of \cite{bhattacharya2023dynamic,behnezhad2023dynamic}. 
Here, a randomized algorithm \emph{works against an adaptive adversary} (or \emph{is adaptive}, for short) if its guarantees hold even when future updates depend on the algorithm's previous output and its internal state. Slightly weaker are \emph{\outputadaptive} algorithms, that allow updates to depend only on the algorithms' output. 
Note that deterministic algorithms are automatically adaptive.
A major motivation to study \outputadaptive dynamic algorithms is their black-box use as subroutines within other algorithms. (See discussions in, e.g., \cite{nanongkai2017dynamic,beimel2022dynamic,chuzhoy2019new}.)

Despite significant effort and success in designing and applying dynamic rounding algorithms, the update time of current $(1-\eps)$-approximate dynamic rounding approaches are slower by large $\poly(\eps^{-1}, \log n)$ factors than an unconditional recourse (changes per update) lower bound of $\Omega(\epsilon^{-1})$ (\Cref{fact:recourse-lb}).\footnote{Proving update time lower bounds for approximate dynamic matching is a notoriously challenging open problem. On the other hand, \cite{solomon2021generalized} show that recourse can be made $O(\eps^{-1})$ for any approximate dynamic matching algorithm.} 
 Consequently, rounding is a computational bottleneck for the running time of many
 state-of-the-art dynamic matching algorithms
\cite{wajc2020rounding,bhattacharya2021deterministic,kiss2022improving,bhattacharya2023dynamic,behnezhad2023dynamic,azarmehr2023fully} and decremental (only allowing deletions) matching algorithms \cite{bernstein2020deterministic,jambulapati2022regularized}.
\newpage

\smallskip 
The question thus arises, 
\emph{can one design (\outputadaptive) optimal dynamic rounding algorithms for fractional matching?} We answer this question in the affirmative in a strong sense.

\subsection{Our Contributions}

Our main results are deterministic and randomized dynamic fractional matching rounding algorithms for bipartite graphs that match the aforementioned simple recourse lower bound of $\Omega(\eps^{-1})$ up to logarithmic factors in $\eps$ and (sub-)logarithmic factors in $n:=|V|$.
These results are  summarized by the following theorem.\footnote{Throughout, we use ``soft-O'' notation, $\tilde{O}$, to suppress logarithmic factors in $\eps$, i.e., $\tilde{O}(f)=O(f\cdot \mathrm{poly}(\log (\eps^{-1})))$.}

\begin{wrapper}
\begin{thm}\label{thm:bipartite}
    The dynamic bipartite matching rounding problem admits:
    \begin{enumerate}
        \item A deterministic algorithm with $\tilde{O}(\eps^{-1} \log n)$ $\update$ time.
        \item An \adaptive algorithm with $\tilde{O}(\eps^{-1} \cdot (\log \log n)^2)$ $\update$ time that is correct w.h.p.
        \item An \outputadaptive algorithm with   $\tilde{O}(\eps^{-1})$ expected $\update$ time.
    \end{enumerate}
  The $\init(G,\x,\eps)$ time of each of these algorithms is $O(\eps\cdot |\supp{x}|)$ times its $\update$ time.
\end{thm}
\end{wrapper}

In contrast, prior approaches have update time at least $\Omega(\eps^{-4})$  (see \Cref{sec:techniques}). 
Moreover, all previous adaptive algorithms with high probability (w.h.p.) or deterministic guarantees all have at least (poly)logarithmic dependence on $n$, as opposed to our (sub-)logarithmic dependence on $n$.

\paragraph{General Graphs.}
In general graphs, one cannot round all fractional matchings (as defined) to integrality while only incurring a $(1-\eps)$ factor loss in value.\footnote{Consider the triangle graph with fractional values $x_e=1/2$ on all three edges; this fractional matching has value $3/2$, while any integral matching in a triangle has size at most one. While adding additional constraints \cite{edmonds1965maximum} avoids this issue, no dynamic fractional algorithms for the matching polytope in general graphs are currently known.}
Nonetheless, 
it is known how to round  ``structured'' $(1/2-\eps)$-approximate dynamic fractional matchings \cite{bhattacharya2017fully,bhattacharya2019deterministically} (see \Cref{def:AMFM}) to obtain an \emph{integral} $(1/2-\eps)$-approximate matching \cite{arar2018dynamic,wajc2020rounding,bhattacharya2021deterministic,kiss2022improving}, and even \emph{almost maximal} matchings \cite{bhattacharya2023dynamic}, as defined in \cite{peleg2016dynamic} and restated below.
	\begin{Def}
 \label{def:amm:main}
		A matching $M$ in $G$ is an \emph{$\eps$-almost maximal matching ($\eps$-AMM)} if $M$ is maximal with respect to some subgraph $G[V\setminus U]$ obtained by removing at most $|U|\leq \eps\cdot \mu(G)$ vertices from $G$, where $\mu(G)$ is the maximum matching size in $G$.
	\end{Def}

    Such $\eps$-AMM's are $(1/2 - \eps)$-approximate with respect to the maximum matching \cite{peleg2016dynamic}. Moreover, (almost) maximality of $\eps$-AMM makes their maintenance a useful algorithmic subroutine \cite{peleg2016dynamic,bhattacharya2023dynamic,azarmehr2023fully}.
    Extending our approach to rounding the aforementioned well-structured, ``near maximal'' dynamic fractional matchings in general graphs \cite{bhattacharya2017fully,bhattacharya2019deterministically}, we obtain faster $\eps$-AMM algorithms, as follows (see \Cref{thm:general-formal} for formal statement).

\begin{wrapper}
    \begin{thm}[Informal version of \Cref{thm:general-formal}]\label{thm:general-informal}
        There exist dynamic algorithms 
        maintaining $\eps$-AMM in general graphs in update time $\tilde{O}(\eps^{-3})+O(t_f + u_f\cdot t_r)$, where $t_f$ and $u_f$ are the update time and number of calls to $\update$ of any  ``structured'' dynamic fractional  matching algorithm, and $t_r$ is the $\update$ time for ``partial'' rounding. Furthermore, there exist dynamic partial rounding algorithms with the same update times and adaptivity as those of \Cref{thm:bipartite}.
    \end{thm}
\end{wrapper}

\subsubsection{Applications}\label{sec:applications}

Applying our rounding algorithms to known fractional algorithms yields a number of new state-of-the-art dynamic matching results.

For example, by a black-box application of \Cref{thm:bipartite}, we deterministically round known decremental (fractional) bipartite matching algorithms \cite{jambulapati2022regularized,bernstein2020deterministic} \emph{with no asymptotic overhead}, 
yielding faster $(1-\eps)$-approximate decremental bipartite matching algorithms.
We also discuss how a variant of \Cref{thm:general-informal} together with the general-graph decremental algorithm of \cite{assadi2022decremental} leads to a conjecture regarding 
the first \emph{deterministic} sub-polynomial update time $(1-\eps)$-approximate decremental matching algorithm in general graphs.

%\footnote{Recently, an $\tilde{O}(\eps^{-1})$-update time decremental algorithm was claimed \cite{zheng2023multiplicative}, but this claim was since retracted on arXiv.}.

Our main application is obtained by applying our rounding algorithm for general graphs of \Cref{thm:general-informal} to the $O(\eps^{-2})$-time fractional matching algorithm of \cite{bhattacharya2019deterministically}, yielding the following.

\begin{restatable}{thm}{AMMresults}\label{thm:AMMresults}
    For any $\eps>0$, there exist dynamic $\eps$-AMM algorithms that are:
    \begin{enumerate}
        \item  Deterministic, using $\tilde{O}(\eps^{-3}\cdot \log n)$  $\update$ time.
        \item \expandafter\capitalize\adaptive, using $\tilde{O}(\eps^{-3} \cdot (\log\log n)^2)$  $\update$ time, correct w.h.p.
        \item \expandafter\capitalize\outputadaptive, using $\tilde{O}(\eps^{-3})$ expected $\update$ time.
    \end{enumerate}
\end{restatable}

In contrast, all prior non-oblivious $(1/2-\eps)$-approximate matching algorithms had at least quartic dependence on $\eps$, which the above result improves to cubic. Moreover, this result yields the first deterministic 
$O(\log n)$-time and adaptive $o(\log n)$-time high-probability algorithms for this widely-studied approximation range and for near-maximal matchings. 
This nearly concludes a long line of work on deterministic/adaptive dynamic matching algorithms for the $(1/2-\eps)$ approximation regime \cite{bhattacharya2015deterministic,bhattacharya2016new,behnezhad2019fully,bhattacharya2019deterministically,wajc2020rounding,bhattacharya2021deterministic,kiss2022improving,bhattacharya2023dynamic}.

\subsection{Our Approach in 
a Nutshell}\label{sec:techniques}

Here we outline our approach, focusing on the key ideas behind \Cref{thm:bipartite}. To better contrast our techniques with those of prior work, we start by briefly overviewing the latter.

\medskip 
\noindent\textbf{Previous approaches.} Prior dynamic rounding algorithms \cite{arar2018dynamic,wajc2020rounding,bhattacharya2021deterministic,kiss2022improving} all broadly work by partially rounding the fractional matching $\x$ to obtain a matching sparsifier $S$ (a sparse subgraph approximately preserving the fractional matching size compared to $\x$).
Then, they periodically compute a $(1-\eps)$-approximate matching in this sparsifier $S$ using a static $\tilde{O}(|S|\cdot \eps^{-1})$-time algorithm (e.g., \cite{duan2014linear}) whenever $\norm{x}$ changes by $\eps\cdot \norm{x}$, i.e., every $\Omega(\eps\cdot \norm{x})$ updates. This period length guarantees that the matching computed remains a good approximation of the current fractional matching during the period, with as good an approximation ratio as the sparsifier $S$.
Now, for sparsifier $S$ to be $O(1)$-approximate, it must have size $|S|=\Omega(\norm{x})$, and so this approach results in an update time of at least $\Omega(\eps^{-2})$. Known dynamic partial rounding approaches all result in even larger sparsifiers, resulting in large $\poly(\eps^{-1}, \log n)$ update times.

\medskip 
\noindent\textbf{Direct to integrality.}
Our first rounding algorithm for bipartite graphs breaks from this framework and directly rounds to integrality. 
This avoids overhead of periodic recomputation of static near-maximum matching algorithms, necessary to avoid super-linear-in-$\eps^{-1}$ update time (or $n^{o(1)}$ factors, if we substitute the static approximate algorithms with the breakthrough near-linear-time max-flow algorithm of \cite{chen2022maximum}).
The key idea is that, by encoding each edge's weight in binary, we can round the fractional matching ``bit-by-bit'', deciding for each edge whether to round a component of value $2^{-i}$ to a component of value $2^{-i+1}$. 
This can be done statically in near-linear-time by variants of standard degree splitting algorithms, decreasing the degree of each node in a multigraph by a factor of two (see \Cref{sec:prelims}).
Letting $L:=\log((\min_{e:x_e\neq 0} x_e)^{-1})$, we show that by buffering updates of total value at most $O(\eps\cdot \norm{x}/L)$ for each power of $2$, we can efficiently dynamize this approach, obtaining a dynamic rounding algorithm with update time $\tilde{O}(\eps^{-1}\cdot L^2)$.
As we can assume that $\min_{e:x_e\neq 0} x_e\geq \eps/n^2$ (\Cref{obs:xmin-and-high-bits}), this gives our bipartite $\tilde{O}(\eps^{-1}\cdot\log^2n)$ time algorithm.

\medskip 
\noindent\textbf{Faster partial rounding.}
The second ingredient needed for \Cref{thm:bipartite} are a number of algorithms for ``partially rounding'' fractional matchings, increasing $\min_{e:x_e\neq 0} x_e$ while approximately preserving the value of the fractional matching. (The output is not quite a fractional matching, but in a sense is close to one. 
See \Cref{def:coarsening}.)
Our partial rounding algorithms draw on a number of techniques, including fast algorithms for partitioning a fractional matching's support into \emph{multiple} sparsifiers, as opposed to a single such sparsifier in prior work, and a new \outputadaptive sampling data structure of possible independent interest (\Cref{sec:app:setsampling}).\footnote{From this we derive the first \outputadaptive matching algorithm that is not also \adaptive.}$^,$\footnote{Concurrently to our work, another such sampling algorithm was devised \cite{yi2023optimal}. See discussion in \Cref{sec:app:setsampling}.}
Combining these partial rounding algorithms with our simple algorithm underlies all our bipartite rounding results of \Cref{thm:bipartite}, as well as our general graph rounding results of \Cref{thm:general-informal} (with further work -- see \Cref{sec:general}).

\subsection{Related Work}\label{sec:related-work}

The dynamic matching literature is vast, and so we only briefly discuss it here. For a more detailed discussion, see, e.g., the recent papers \cite{behnezhad2023dynamic,bhattacharya2023dynamic,bhattacharya2023dynamicsublinear,assadi2023regularity}.

The dynamic matching problem has been intensely studied since a seminal paper of Onak and Rubinfeld \cite{onak2010maintaining}, which showed how to maintain a constant-approximate matching in polylogarithmic time.
Results followed in quick succession, including conditional polynomial update time lower bounds for exact maximum matching size \cite{abboud2014popular,henzinger2015unifying,abboud2016popular,dahlgaard2016hardness,kopelowitz2016higher}, and numerous algorithmic results, broadly characterized into two categories:
polynomial time/approximation tradeoffs \cite{gupta2013fully,peleg2016dynamic,bernstein2015fully,bernstein2016faster,van2019dynamic,grandoni2022maintaining,kiss2022improving,behnezhad2020fully,wajc2020rounding,bhattacharya2021deterministic,behnezhad2022new,roghani2022beating,behnezhad2023dynamic,assadi2023regularity,bhattacharya2023dynamicsublinear,ivkovic1993fully}, and  $1/2-$ or $(1/2-\eps)$-approximate
algorithms with polylogarithmic or even constant update time \cite{baswana2015fully,solomon2016fully,bhattacharya2016new,arar2018dynamic,charikar2018fully,bernstein2019deamortization,bhattacharya2019deterministically,behnezhad2019fully,chechik2019fully,wajc2020rounding,bhattacharya2021deterministic,kiss2022improving}.\footnote{Some works study  approximation of maximum matching \emph{size} \cite{sankowski2009maximum,van2019dynamic,behnezhad2023dynamic,bhattacharya2023dynamic,bhattacharya2023dynamicsublinear,azarmehr2023fully}.}
We improve the state-of-the-art update times for all deterministic and adaptive algorithms in the intensely-studied second category.

The $(1-\eps)$-approximate matching problem has also been studied in \emph{partially dynamic} settings. 
This includes a  recent algorithm supporting vertex updates on opposite sides of a bipartite graph, though not edge updates \cite{zheng2023multiplicative} (see arXiv).
For incremental (edge-insertion-only) settings several algorithms are known \cite{gupta2014maintaining,grandoni20191+,blikstad2023incremental,bhattacharya2023dynamicLP}, the fastest  having
$\poly(\eps^{-1})$ update time \cite{blikstad2023incremental}. In decremental settings (edge-deletion-only), rounding-based algorithms with $\poly(\eps^{-1},\log n)$ update time in bipartite graphs \cite{bernstein2020deterministic,jambulapati2022regularized,bhattacharya2023dynamicLP} and randomized $\exp(\eps^{-1})$ in general graphs \cite{assadi2022decremental} are known.
We improve on these decremental results, speeding up bipartite matching, and giving the first deterministic logarithmic-time  algorithm for general graphs.

\subsection{Paper Outline}
Following some preliminaries in \Cref{sec:prelims}, we provide our first simple bipartite rounding algorithm in \Cref{sec:sequential}. 
In \Cref{sec:partial-rounding} we introduce the notion of partial roundings that we study and show how such partial rounding algorithms (which we provide and analyze in \Cref{sec:partial-rounding-implementations}) can be combined with our simple algorithm to obtain the (bipartite) rounding algorithms of \Cref{thm:bipartite}.
In \Cref{sec:general} we analyze our rounding algorithms when applied to known fractional matchings in general graphs, from which we obtain \Cref{thm:general-informal,thm:AMMresults}.
We conclude with our decremental results in \Cref{sec:decremental}.
\section{Preliminaries}\label{sec:prelims}

\noindent\textbf{Assumptions and Model.} Throughout, we assume that $\norm{x}\geq 1$, as otherwise it is trivial to round $\norm{x}$ within a factor of $1-\eps$, by maintaining a pointer to any edge in $\supp{x}$ whenever the latter is not empty.
In this paper we work in the word RAM model of computation with words of size $w:=\Theta(\log n)$, allowing us to index any of $2^{O(w)}=\poly(n)$ memory addresses, perform arithmetic on $w$-bit words, and draw $w$-bit random variables, all in constant time. 
We will perform all above operations on $O(\log(\eps^{-1}\cdot n))$-bit words, which is still $O(w)$ provided $\eps^{-1}=\poly(n)$. If $\eps$ is much smaller, all stated running times trivially increase by a factor of $O(\log(\eps^{-1}))$.

\vspace{-0.375cm}
\paragraph{Notation.} For multisets $S_1$ and $S_2$, we 
denote by  $S_1\uplus S_2$ the ``union'' multiset, in which each element has multiplicity that is the sum of its multiplicities in $S_1$ and $S_2$. 
A vector $\x$ is {\em $\lambda$-uniform} if $x_e = \lambda$ for all $e \in \supp{\x}$, and is \emph{uniform} if it is $\lambda$-uniform for some $\lambda$. Given fractional matching $\x$, we call an integral matching $M \subseteq \supp{\x}$ that is $(1-\eps)$-approximate, i.e., $|M| \geq \norm{\x} \cdot (1 - \eps)$ an \emph{$\eps$-approximate rounding} of $\x$.
Finally, we use the following notion of distance and its monotonicity.\vspace{-0.5cm}
\begin{obs}\label{obs:monotone-distance}
For vectors $\x,\y\in \mathbb{R}^{E}$ and $\eps\geq 0$, define $d^\epsilon_V(\x, \y) := \sum_{v \in V} (|x(v) - y(v)| - \epsilon)^+$, for $(z)^+:=\max(0, z)$ the positive part of $z\in \mathbb{R}$. Then, we have 
 $d^{\eps}_V(\x,\y)\leq d^{\eps'}_V(\x,\y)$ for all $\epsilon\geq \epsilon'$. Moreover, by the triangle inequality and the basic fact that $(a+b)^+\leq a^+ + b^+$ for all real $a,b$, we have $d^{\eps_1+\eps_2}_V(\x,\z)\leq d^{\eps_1}_V(\x,\y) + d^{\eps_2}_V(\y,\z)$ for all $\eps_1,\epsilon_2\geq 0$ and vectors $\x,\y,\z\in \mathbb{R}^E$.
\end{obs}
\vspace{-0.375cm}
\paragraph{Support and Binary encoding.}
We denote the binary encoding of each edge $e$'s fractional value by $x_e \defeq \sum_i (x_e)_i \cdot  2^{-i}$.
We further let $\suppOp_i(\x)\defeq \{e\in E \mid (x_e)_i=1\}$ denote the set of coordinates of $\x$ whose $i$-th bit is a $1$. So, $\suppOp(\x)=\bigcup_i \suppOp_i(\x)$.
Next, we let $\x_{\min}\defeq \min_{e \in \support{x}} x_e$.
The following observation allows us to restrict our attention to a small number of bits when rounding bipartite fractional matchings $\x$. (In \Cref{sec:general} we extend this observation to the structured fractional matchings in general graphs that interest us there.)

\begin{obs}\label{obs:xmin-and-high-bits}
For rounding bipartite fractional matching, by decreasing $\eps$ by a constant factor, it is without loss of generality that $\x_{\min}\geq \eps/n^2$ and moreover if $\Delta\leq \x_{\min}$ and $L:=1+\lceil \log(\eps^{-1}\Delta^{-1})\rceil$, we may safely assume that $(x_e)_i=0$ for all $i>L$.
\end{obs}
\begin{proof}
Let $\eps'=\eps/3$. Consider the vector $\x'$ obtained by zeroing out all entries $e$ of $\x$ with $x_e<\eps'/n^2$ and setting $(x_e)_i=0$ for all edges $e$ and indices $i>L$. 
Clearly, $\supp{x'}\subseteq \supp{x}$ and $\x'$ is a fractional matching, as $\x'\leq \x$.
The following shows that $\norm{x'}$ is not much smaller than $\norm{x}\geq 1$.
$$\norm{x'} \geq \norm{x} - {n \choose 2}\eps'/n^2 - \sum_{e}\sum_{i>L} 2^{-i} \geq \norm{x} - \eps' - \sum_e \eps'\cdot \x_{\min} \geq \norm{x}\cdot (1-\eps') - \sum_e \eps'\cdot x_e = (1-2\eps')\cdot \norm{x}.$$
Therefore, a matching $M\subseteq \supp{x'}\subseteq \supp{x}$ that is $(1-\eps')$-approximate w.r.t.~$x'$ is $(1-\eps)$-approximate w.r.t.~$\x$, as $|M|\geq (1-\eps')\cdot \norm{x'}\geq (1-3\eps')\cdot \norm{x} = (1-\eps)\cdot \norm{x}$.
\end{proof}

\vspace{-0.375cm}
\paragraph{Recourse Lower Bound.}
We note that the number of changes to $M$ 
per update (a.k.a.~the rounding algorithm's \emph{recourse}) is at least $\Omega(\eps^{-1})$ in the worst case.

\begin{fact}\label{fact:recourse-lb}
Any $(1-\eps)$-approximate dynamic matching rounding algorithm $\calA$ must use $\Omega(\eps^{-1})$ amortized recourse, even in bipartite graphs.
\end{fact}
\begin{proof}
Consider a path graph $G$ of odd length $4\eps^{-1}+2$ with values $1/2$ assigned to each edge. 
A matching  $M\subseteq  \supp{x}$ of size  $|M|\geq (1-\eps)\cdot \norm{x}$ must match all odd-indexed edges of the path.
However, after invoking $\update(\cdot ,0)$ for the first and last edges in the path, for $|M|\geq (1-\eps)\cdot \norm{x}$ (for the new $\x$),  $M$ must match all \emph{even}-indexed edges. Therefore, repeatedly invoking $\update(\cdot, 0)$ and then $\update(\cdot, 1/2)$ for these two edges sufficiently many times implies that the matching $M$ maintained by $\calA$ must change by an average of $\Omega(\eps^{-1})$ edges per update.
\end{proof}

    \subsection{The $\dsplit$ subroutine}

	Throughout the paper, 
    we use the following subroutine to partition a graph into two subgraphs of roughly equal sizes while roughly halving all vertices' degrees. 
    Such subroutines obtained by e.g., computing maximal walks and partitioning them into odd/indexed edges, have appeared in the literature before in various places. For completeness, we provide this algorithm in \Cref{app:dsplit}.

    \begin{restatable}{prop}{dsplitalgo}\label{prop:degree-split}
        There exists an algorithm $\dsplit$, which on multigraph $G=(V,E)$ with maximum edge multiplicity at most two (i.e., no edge has more than two copies) computes in $O(|E|)$ time two (simple) edge-sets $E_1$ and $E_2$ of two disjoint sub-graphs of $G$, such that $E_1,E_2$ and the degrees $d_G(v)$ and $d_i(v)$ of $v$ in $G$ and $H_i:=(V,E_i)$ satisfy the following.
        \begin{enumerate}[label=(P{{\arabic*}})]
            \item \label{property:size-halved}$|E_1|=\lceil \frac{|E|}{2}\rceil$ and $|E_2|= \lfloor \frac{|E|}{2}\rfloor$.
            \item \label{property:degree-halved} $d_i(v)\in   \left[ \frac{d_G(v)}{2}-1,\,\frac{d_G(v)}{2}+1\right]$ for each vertex $v\in V$ and $i\in \{1,2\}$.
        \item \label{property:bipartite-degree-halved}  $d_i(v)\in \left[\lfloor\frac{d_G(v)}{2}\rfloor,\ \lceil\frac{d_G(v)}{2}\rceil\right]$
        for each vertex $v\in V$ and $i\in \{1,2\}$ \underline{if $G$ is bipartite}.
        \end{enumerate}
    \end{restatable}
	\section{Simple Rounding for Bipartite Matchings}\label{sec:sequential}

    In this section we use the binary encoding of $\x$ to approximately round fractional bipartite matchings in a ``linear'' manner, rounding from the least-significant to most-significant bit of the encoding. We first illustrate this approach in a static setting in \Cref{sec:warmup}. This will serve as a warm-up for our first dynamic rounding algorithm provided in \Cref{sec:bip-dynamic}, which is essentially a dynamic variant of the static algorithm (with 
    its $\init$ procedure being essentially the static algorithm).
	
	\subsection{Warm-up: Static Bipartite Rounding} \label{sec:warmup}

In this section, we provide a simple static bipartite rounding algorithm for fractional matchings. 

Specifically, we prove the following \Cref{thm:static_bipartite_round}, analyzing our rounding algorithm, \Cref{alg:rounding_hierarchy}. The algorithm simply considers for all $i$, $E_i := \suppOp_i(x)$, i.e., the edges whose $i$-th bit is set to one in $\x$. Starting from  $F_L=\emptyset$, for $i=L,\dots,1$, the algorithm applies $\dsplit$ to the multigraph $G[F_i\uplus E_i]$ and sets $F_{i-1}$ to be the first edge-set output by $\dsplit$ (by induction, $E_i,F_i$ are simple sets, and so $G[F_i\uplus E_i]$ has maximum multiplicity two.)
Overloading notation slightly, we denote this by $F_{i -1} \gets \dsplit(G[F_i\uplus E_i])$. The algorithm then outputs $E_0\cup F_0$.

\begin{algorithm}
\caption{Hierarchical Fractional Rounding Algorithm}
\label{alg:rounding_hierarchy}
\SetKwInOut{Input}{input}
\SetKwInOut{Output}{output}

\Input{Fractional matching $\x \in \R^E_{\geq 0}$ in graph $G = (V,E)$}
\Input{Accuracy parameter $\epsilon \in (0, 1)$}
\Output{Integral matching $M \subseteq \support{x}$ with $|M| \geq (1 - \epsilon)\cdot \norm{x}$}

\BlankLine
$L \gets 1+\lceil \log_2(\eps^{-1} x_{\min}^{-1})\rceil$ and $F_L \gets \emptyset$\;
\For{$i = L, L - 1, \ldots, 1$}{
    $E_i\gets \suppOp_i(\x)$\;
    $F_{i-1} \gets \dsplit(G[E_i\uplus F_i])$ \tcp*{First set output by $\dsplit$} 
}
\textbf{return} $M \gets E_0\cup F_0$\;
\end{algorithm}

\begin{thm}\label{thm:static_bipartite_round}
On fractional bipartite matching $\x$ and error parameter $\epsilon \in (0, 1)$,~\Cref{alg:rounding_hierarchy} outputs an integral matching $M \subseteq \support{x}$ with $|M| \geq (1 - \epsilon) \cdot \norm{\x}$ in time $O(|\supp{\x}| \cdot \log (\eps^{-1} \x_{\min}^{-1}))$.
\end{thm}

By \Cref{obs:xmin-and-high-bits}, $L=O(\log(\eps^{-1}\cdot n))$, and so \Cref{thm:static_bipartite_round} implies an $O(|\supp{x}|\cdot \log (\eps^{-1}\cdot n))$ runtime for \Cref{alg:rounding_hierarchy}.
    We prove this theorem in several steps. Key to our analysis is the following sequence of vectors (which we will soon show are fractional matchings if $\supp{x}$ is bipartite).

    \begin{Def}
    Letting $F_i(e):=\mathds{1}[e\in F_i]$ and $E_i(e):=\mathds{1}[e\in E_i] = (x_e)_i$, we define a sequence of vectors $\x^{(i)}\in \mathbb{R}^E_{\geq 0}$ for $i=0,1,\dots,L$ as follows.
\begin{align}\label{rounded-frac}
	x^{(i)}_e := F_i(e)\cdot 2^{-i} + \sum_{j = 0}^{i} E_j(e) \cdot 2^{-j}.
	\end{align}
 \end{Def}
	So, $\normOp{\x^{(L)}}\geq (1-\eps)\cdot \norm{\x}$, by definition and \Cref{obs:xmin-and-high-bits}. Moreover, each (copy of) edge $e$ output/discarded by $\dsplit(G[E_i\uplus F_i])$ corresponds to adding/subtracting $2^{-i}$ to/from $x^{(i)}_e$ to obtain $x^{(i-1)}_e$. 
    This allows us to prove the following lower bound on the size of the output.
    \begin{lem}\label{lem:sequential-flow-preserved}
    $\normOp{\x^{(i)}}\geq (1-\eps)\cdot \norm{x}$ for all $i\in \{0,1,\dots,L\}$.
    \end{lem}
\begin{proof}
    By Property \ref{property:size-halved} we have that $|F_{i-1}|\geq \lceil \frac{1}{2} (|F_i|+|E_{i}|)\rceil$ and so \begin{align*}
        \normOp{\x^{(i-1)}} & = \normOp{\x^{(i)}} + 2^{1-i}\cdot \sum_e F_{i-1}(e) - \sum_e 2^{-i}\cdot (F_i(e)+E_i(e))\geq \normOp{\x^{(i)}}.
    \end{align*}
    Therefore, repeatedly invoking the above bound and appealing to \Cref{obs:xmin-and-high-bits}, we have that indeed, for all $i\in \{0,1,\dots,L\}$,
    \begin{align*}
    \normOp{\x^{(i)}} & \geq \normOp{\x^{(i+1)}} \geq \dots \geq \normOp{\x^{(L)}} \geq (1-\eps)\cdot \norm{x}. \qedhere
	\end{align*}
\end{proof}
	A simple proof by induction shows that if $\supp{x}$ is bipartite, then the above procedure preserves all vertices' fractional degree constraints, i.e., the vectors $\x^{(i)}$ are all fractional matchings.
	\begin{lem}\label{lem:static-bip-feasible}
	If $\x$ is a fractional bipartite matching then $x^{(i)}(v) \leq 1$ for every vertex $v\in V$ and $i\in \{0,1,\dots,L\}$.
	\end{lem}
\begin{proof}
	By reverse induction on $i\leq L$. The base case holds since $\x^{(L)}\leq \x$ is a fractional matching.
	To prove the inductive step for $i-1$ assuming the inductive hypothesis $x^{(i)}(v)\leq 1$,
	let $d_{G_i}(v)=\sum_{e\in v} \left(E_i(e)+F_i(e)\right)$
	be the number of (possibly parallel) edges incident to $v$ in $G_i:=G[E_i\uplus F_i]$.
	By Property \ref{property:bipartite-degree-halved}, 
    we have the following upper bound on $v$'s fractional degree under $\x^{(i-1)}$.
	\begin{align}\label{eqn:bip-static-new-degrees}
	x^{(i-1)}(v) \leq x^{(i)}(v) - d_{G_i}(v)\cdot 2^{-i} + \left\lceil \frac{d_{G_i}(v)}{2}\right\rceil \cdot 2^{-i+1}.
	\end{align}
	If $d_{G_i}(v)$ is even, then we are done, by the inductive hypothesis giving $x^{(i-1)}(v)\leq x^{(i)}(v) \leq 1$. Suppose therefore that $d_{G_i}(v)$ is odd.
	By \Cref{rounded-frac}, any value $x^{(i)}_e$ is evenly divisible by $2^{-i}$ and therefore the same holds for $x^{(i)}(v)$. By the same token, $d_{G_i}(v)$ is odd if and only if $x^{(i)}(v)$ is not evenly divisible by $2^{-i+1}$. However, since $x^{(i)}(v)$ is evenly divisible by $2^{-i}$ and it is at most one, this implies that $x^{(i)}(v)\leq 1-2^{-i}$. Combined with \Cref{eqn:bip-static-new-degrees}, we obtain the desired inequality when $d_{G_i}(v)$ is odd as well, since \begin{align*}
    x^{(i-1)}(v) & \leq x^{(i)}(v) + 2^{-i} \leq  1-2^{-i}+2^{-i} = 1.\qedhere
    \end{align*}
\end{proof}
	
	Now, since the vector $\x^{(0)}$ is integral, the preceding lemmas imply that if $\x$ is a bipartite fractional matching then $M$ is a large \emph{integral} matching.
    \begin{lem}
        If $\x$ is a fractional bipartite  matching then $M=E_0\cup F_0\subseteq \supp{x}$ is an integral matching of cardinality at least $|M|=|E_0|+|F_0|\geq (1-\eps)\cdot \norm{x}$.
    \end{lem}
\begin{proof}
        By \Cref{lem:static-bip-feasible}, the (binary) vector $\x^{(0)}$ (the characteristic vector of $M$) is a feasible fractional matching, and so $M$ is indeed a matching.
        That $M\subseteq \supp{x}$ follows from $\dsplit$ outputting a sub(multi)set of the edges of its input, and therefore a simple proof by induction proves that $\suppOp(\x) \supseteq  \suppOp(\x^{(L)}) \supseteq \suppOp(\x^{(L-1)}) \supseteq \dots  \supseteq \suppOp(\x^{(0)}) = M$. The lower bound on $|M|=\normOp{x^{(0)}}$ then follows from \Cref{lem:sequential-flow-preserved}.
\end{proof}

    Finally, we bound the algorithm's running time.
    \begin{lem}
        \Cref{alg:rounding_hierarchy} takes time $O(|\supp{x}|\cdot L)$ when run on vector $\x\in \mathbb{R}^{E}_{\geq 0}$.
    \end{lem}
\begin{proof}
    To analyze the runtime of the algorithm, note that it runs in time $O(L + \sum_{i = 0}^{L} (|F_i| + |E_i|))$. Further, $|F_L| = 0$ and
    by Property \ref{property:size-halved} we have that $|F_i| \leq \frac{1}{2} |F_{i + 1}| + \frac{1}{2}|E_{i + 1}| + 1$ for all $i \in \{0,1,\dots,L-1\}$. Letting $m \defeq |\support{x}|$ we know that $|E_i| \leq m$ for all $i$, and so by induction
    \[
    |F_i| \leq \frac{1}{2}|F_{i+1}|+\frac{1}{2}m + 1
    \leq m + 2
    \qquad \text{ for all }
    j \in \{0, 1, ... , L - 1\}.
    \]
    Thus, the algorithm runs in the desired time of $O(m L + L) = O(m L)$.
\end{proof}
    \Cref{thm:static_bipartite_round} follows by combining the two preceding lemmas. 

    We now turn to the dynamic counterpart of \Cref{alg:rounding_hierarchy}.
	
	\subsection{A Simple Dynamic Bipartite Rounding Algorithm}\label{sec:bip-dynamic}

    In this section we dynamize the preceding warm-up static algorithm, obtaining the following result.

    \begin{thm}\label{thm:sequential-algo}
    \Cref{alg:dynamic_rounding_hierarchy} is a deterministic dynamic bipartite matching rounding algorithm. Under the promise that the dynamic input vector $\x$ satisfies $\x_{\min}\geq \delta$ throughout, its amortized $\update$ time is $O(\eps^{-1} \cdot \log^2(\eps^{-1}\delta^{-1}))$ and  its $\init$ time on vector $\x$ is $O(|\supp{x}|\cdot\log(\eps^{-1} \delta^{-1}))$.
    \end{thm}

    Since $\delta\geq \eps/n^2$ by \Cref{obs:xmin-and-high-bits}, \Cref{thm:sequential-algo} yields an $\tilde{O}(\eps^{-1}\cdot \log^2n)$ update time algorithm.

    Our dynamic algorithm follows the preceding static approach. For example, its initialization is precisely the static \Cref{alg:rounding_hierarchy} (and so the $\init$ time follows from \Cref{thm:static_bipartite_round}). In particular, the algorithm considers a sequence of graphs $G_i:=G[E_i\uplus F_i]$\ and fractional matchings $\x^{(i)}$ defined by $G_i$ and the $i$ most significant bits of $x_e$, as in \Cref{rounded-frac}.
	However, to allow for low (amortized) update time we allow for a small number of unprocessed changed or deleted edges for each $i$, denoted by $c_i$. 
    When such a number $c_i$ becomes large, we \emph{rebuild} the solution defined by $F_i$ and $\suppOp_i(\x),\dots,\suppOp_0(\x)$ as in the static algorithm.
    Formally, 
    our algorithm is given in \Cref{alg:dynamic_rounding_hierarchy}.

\begin{algorithm}[t]
\caption{Hierarchical Dynamic Fractional Rounding Algorithm}
\label{alg:dynamic_rounding_hierarchy}
\SetKwInOut{State}{global}
\SetKwInOut{Return}{return}
\SetKwProg{Fn}{function}{}{}

\State{Current vector $\x$}
\State{Current output integral matching  $M$}
\State{Accuracy $\epsilon$ and maximum layer $L \in \mathbb{Z}_{> 0}$}
\State{Partial roundings $F_0,F_1, \ldots, F_{L} \subseteq E$ and update counts $c_0, c_1, \ldots, c_L \in  \mathbb{Z}_{\geq 0}$
}

\BlankLine
\tcp{In $\init$ we assume that the algorithm knows $\delta$, a lower bound on $\x_{\min}$ for all nonzero $\x$ encountered after an operation}
\Fn{$\init(G = (V,E),\; \x \in \R^E_{\geq 0}, \; \epsilon \in (0,1))$}{
    Save $\x$ and $\epsilon$ as global variables\;
    Initialize $L \gets 1+\lceil \log_2(\eps^{-1} \delta^{-1}) \rceil$, $c_i \gets 0$, and $F_i \gets \emptyset$, for all $i \in \{0,1,\ldots,L\}$\;
    Call $\rebuild(L)$\;
}

\BlankLine

\Fn{$\rebuild(i)$}{
    \For{$j = i, i - 1, \ldots, 0$}{
        $E_j\gets \suppOp_j(\x)$ and $c_j\gets 0$\;
        \lIf{$j \neq 0$}{$F_{j-1} \gets \dsplit(G[E_j\uplus F_j])$}
    }
    $M \gets E_0\cup F_0$\;
}

\BlankLine

\Fn{$\update(e,\nu)$}{
    $\x_e \gets \nu$\;
    \For{$i = L, L - 1,...,0$}{
        \lIf{$e \in E_i$}{remove $e$ from $E_i$}
       
        \If{$i \neq 0$}{
            \lIf{$e \in F_{i-1}$}{remove $e$ from $F_{i-1}$}
            \lElse{remove one edge adjacent to each endpoint of $e$ from $F_{i-1}$ (if there is one)}
        }
        $c_i\gets c_i+1$\;
        \lIf{$c_i > 2^{i-2} \cdot \frac{\eps\norm{x}}{L}$}{call $\rebuild(i)$ and \textbf{return}} 
    }
}

\end{algorithm}

    \medskip\noindent\textbf{Conventions and notation.}
    Most of our lemmas concerning \Cref{alg:dynamic_rounding_hierarchy} hold for arbitrary non-negative vectors $\x\in \mathbb{R}^{E}_{\geq 0}$, a fact that will prove useful in later sections.
    We state explicitly which lemmas hold if $\x$ is a fractional bipartite matching.
    In the analysis of \Cref{alg:dynamic_rounding_hierarchy} we let $\x^{(i)}$ be as defined in \Cref{rounded-frac}, but with $E_i$ and $F_i$ of the dynamic algorithm. 
    Furthermore, we prove all structural properties of \Cref{alg:dynamic_rounding_hierarchy} for any time after $\init$ and any number of $\update$ operations, and so we avoid stating this in all these lemmas' statements for brevity. Next, we use the shorthand $S_i:=\suppOp_i(\x)$, and note that unlike in the static algorithm, due to deletions from $E_i$ before the next $\rebuild(i)$, the containment $E_i\subseteq S_i$ may be strict.
    \medskip

    First, we prove that $M$ is a matching if $\x$  is a bipartite fractional matching. More generally, we prove that each $\x^{(i)}$, and in particular $\x^{(0)}$, is a fractional matching, implying the above.
    
    \begin{lem}\label{lem:feasible-fractional}
    If $\x$ is a fractional bipartite matching, then $\x^{(i)}$ is a fractional matching for all $i\in \{0,1,\dots, L\}$.
    \end{lem}

\begin{proof} 
    Fix vertex $v$, and let $F_i(v)$ and $S_i(v)$ be the number of edges of $v$ in $F_i$ and $S_i$ respectively, for all $i\in \{0,1,\dots, L\}$.
    To upper bound $x^{(i)}(v)$, we start by upper bounding $F_i(v)$, as follows.
	\begin{align} \label{eqn:fiv-upper-bound}
	F_{i}(v) & \leq \left\lceil 2^{i} \cdot \sum_{j=i+1}^L S_j(v)\cdot 2^{-j}\right\rceil\,.
	\end{align}
    We prove the above by induction on the number of operations and by reverse induction on $i\in \{0,1,\dots,L\}$, as follows. The base case $i=L$ is trivial, as $F_L(v)=0$ throughout and the RHS is non-negative.
    Next, for $i<L$, consider the effect on $F_{i}(v)$ of an update resulting in a call to $\rebuild(i+1)$ (e.g., after calling $\init$), at which point $E_{i+1}\gets S_{i+1}$.
    \begin{align*}
        F_i(v) & \leq \left\lceil \frac{1}{2}\cdot (S_{i+1}(v) + F_{i+1}(v)) \right\rceil  & \textrm{Property \ref{property:bipartite-degree-halved}} \\
	   & \leq \left\lceil \frac{1}{2}\cdot S_{i+1}(v) + \frac{1}{2}\cdot \left\lceil 2^{i+1} \cdot \sum_{j=i+2}^L S_j(v)\cdot 2^{-j}\right\rceil \right\rceil  & \textrm{Inductive hypothesis for $i+1$}\\ 
	   & \leq \left\lceil \frac{1}{2}\cdot S_{i+1}(v) + \frac{1}{2}\cdot  2^{i+1} \cdot \sum_{j=i+2}^L S_j(v)\cdot 2^{-j}\right\rceil \\ 
        & = \left\lceil 2^{i} \cdot \sum_{j=i+1}^L S_j(v)\cdot 2^{-j}\right\rceil \,,
    \end{align*}
	where the last inequality follows from the basic fact that for non-negative $y,z$ with $y$ an integer, $\lceil \frac{1}{2}\cdot y + \frac{1}{2} \lceil z \rceil\rceil \leq \lceil \frac{1}{2}(y+ z)\rceil$.
    Next, it remains to prove the inductive step for index $i$ and a call to $\update$ for which $\rebuild(i+1)$ is not called: but such an update only decreases the left-hand side of \Cref{eqn:fiv-upper-bound}, while it causes a decrease in the right-hand side (by one) only if an edge of $v$ was updated in this call to $\update$, in which case we delete at least one edge incident to $v$ in $F_i$, if any exist, and so the left-hand side also decreases by one (or is already zero). 

    Finally, combining \Cref{rounded-frac} and \Cref{eqn:fiv-upper-bound}, we obtain the desired inequality $x^{(i)}(v)\leq 1$.
    \begin{align*}
    x^{(i)}(v) & \leq  F_{i}(v)\cdot 2^{-i} + \sum_{j=0}^{i}S_j(v)\cdot 2^{-j} 
    \leq  2^{-i} \cdot \left\lceil 2^i\cdot \sum_{j=i+1}^L S_j(v)\cdot 2^{-j}\right\rceil + \sum_{j=0}^{i}S_j(v)\cdot 2^{-j}
     \leq 1.
    \end{align*}
    Above, the first inequality follows from \Cref{rounded-frac} and $E_i(v)\leq S_i(v)$ since $E_i\subseteq S_i$. The second inequality follows from \Cref{eqn:fiv-upper-bound}.
    Finally, the final inequality relies on $\sum_j S_{j}(v)\cdot 2^{-j} = x(v) \leq 1$, together with $2^i\cdot \sum_{j=i+1}^L S_j(v)\cdot 2^{-j}$ being fractional if and only if $\sum_{j=i+1}^L S_j(v)\cdot 2^{-j}$ is not evenly divisible by $2^{-i}$, though it is evenly divisible by $2^{-i-1}$, in which case $x(v)\leq 1-2^{-i}$.
\end{proof}

    \begin{rem}\label{rem:Fi-ub}
    The same proof approach, using Property \ref{property:size-halved} of $\dsplit$ (for possibly non-bipartite graphs) implies the global upper bound $|F_{i}| \leq \left\lceil 2^{i} \cdot \sum_{j=i+1}^L |S_j|\cdot 2^{-j}\right\rceil \leq 1 + 2^i\cdot \norm{x}$.
    \end{rem}

    Next, we prove the second property of a rounding algorithm, namely that $M\subseteq \supp{x}$.
    \begin{lem}\label{lem:sequential-containment}
$\suppOp{(\x^{(i)})}\subseteq \supp{x}$ for all $i\in \{0,1\dots,L\}$ and therefore 
       $M\subseteq \supp{x}$.
    \end{lem}
\begin{proof}
        We prove the stronger claim by induction on the number of operations of \Cref{alg:dynamic_rounding_hierarchy} and by reverse induction on $i\in \{0,1,\dots,L\}$ that $\suppOp{(\x^{(i)})} = E_i\cup F_i\subseteq \suppOp{(\x^{(i+1)})} \subseteq \supp{x}$. 
        That $E_i\subseteq \suppOp_i(\x)$ throughout is immediate, since $E_i\gets \suppOp_i(\x)$ when $\rebuild(i)$ is called (and in particular after $\init$ was called), and subsequently all edges $e\in E_i$ that are updated (and in particular each edge whose $x_e$ value is set to zero) are removed from $E_i$. 
        Therefore $E_i\subseteq \supp{x}$ throughout, and in particular $\suppOp{(\x^{(L)})} = E_L \subseteq \supp{x}$.
        Similarly, by the properties of $\dsplit$ and the inductive hypothesis, we have that after $\rebuild(i)$ is called, $F_{i}\subseteq F_{i+1} \cup E_{i+1}\subseteq \suppOp{(\x^{(i+1)})} \subseteq \supp{x}$, and each edge $e$ updated since is subsequently deleted from $F_i$ (as are some additional edges).
        Therefore $F_i\subseteq \supp{x}$ throughout.
        We conclude that $\suppOp{(\x^{(i)})} \subseteq \supp{x}$ for all $i$, as desired.
\end{proof}

    We now argue that the unprocessed edges have a negligible effect on values of $\x^{(i)}$ compared to their counterparts obtained by running the static algorithm on the entire input $\x$. 

    \begin{lem}\label{lem:sequential-size} 
    $\normOp{\x^{(i)}} \geq (1-2\eps)\cdot \norm{x}$ for every $i\in \{0,1,\dots,L\}$.
    \end{lem}
\begin{proof}
    As in the proof of \Cref{lem:sequential-flow-preserved}, by Property \ref{property:size-halved}, after $\init$ or any $\update(\cdot,\cdot)$ triggering a call to $\rebuild(i)$ we have that $|F_{i-1}| = \lceil \frac{1}{2} (|F_i|+|E_{i}|)\rceil$, and so $\normOp{\x^{(i-1)}} \geq \normOp{\x^{(i)}}$. 
    On the other hand, between calls to $\rebuild(i)$ there are at most $2^{i-2}\cdot \frac{\eps \norm{x}}{L}$ calls to $\update(e,\nu)$, resulting in at most $2^{i-1}\cdot \frac{\eps \norm{x}}{L}$ many edges being deleted from $F_{i-1}$, which in turn result in $\normOp{\x^{(i-1)}}$ decreasing by at most $2^{-(i-1)}\cdot 2^{i-1}\cdot \frac{\eps \norm{x}}{L} = \frac{\eps\norm{x}}{L}$.
    In contrast, by \Cref{rounded-frac}, any changes in $E_j$ for $j\neq i$ have no effect on $\normOp{\x^{(i-1)}}-\normOp{\x^{(i)}}$. 
    On the  other hand, until the next $\rebuild(i)$ is triggered, we have that $E_i$ and $F_i$ can only decrease (contributing to an increase in $\normOp{x^{(i-1)}} - \normOp{x^{(i)}}$); $E_i$ can only decrease since edges are only added to $E_i$ when $\rebuild(i)$ is called, and $F_i$ only decreases until  $\rebuild(i+1)$ is called, which triggers a call to $\rebuild(i)$.
    Therefore, $\normOp{x^{(i-1)}} - \normOp{x^{(i)}}$ decreases by at most $\frac{\eps\norm{x}}{L}$ during updates until the next call to $\rebuild(i)$, and so after $\init$ and after every $\update$ of \Cref{alg:dynamic_rounding_hierarchy}, we have that
    $$\normOp{\x^{(i-1)}}\geq \normOp{\x^{(i)}} - \frac{\eps\norm{x}}{L}.$$ 
    Invoking the above inequality $L$ times, and using that $\normOp{\x^{(L)}}\geq (1-\eps)\cdot \norm{x}$ 
    by \Cref{obs:xmin-and-high-bits}, we obtain the desired inequality.
    \begin{align*}
    	|E_0\cup F_0| = \normOp{\x^{(0)}} & \geq \normOp{\x^{(1)}} - \frac{\eps \norm{x}}{L} \geq \dots \geq \normOp{\x^{(L)}} - L\cdot \frac{\eps \norm{x}}{L} \geq (1-2\eps)\cdot\norm{x}. \qedhere 
     \end{align*}
\end{proof}

\begin{rem}\label{lem:sequential-size-upper-bound} 
	The latter is nearly tight, as 	$\normOp{\x^{(i)}} \leq  (1+\epsilon)\cdot \norm{x} + 2^{1-i}$ for every $i\in \{0,1,\dots,L\}$.
\end{rem}
\begin{proof}[Proof (Sketch)] 
	The proof follows that of \Cref{lem:sequential-size}, with the following changes:
	By Property \ref{property:size-halved}, after $\init$ or any $\update(\cdot,\cdot)$ triggering a call to $\rebuild(i)$ we have the upper bound $|F_{i-1}| = \lceil \frac{1}{2} \leq 1+\frac{1}{2}(|F_i|+|E_i|)$, and so $\normOp{\x^{(i-1)}}\leq \normOp{\x^{(i)}}+2^{-i+1}$. 
	 On the other hand, the increase in $\normOp{x^{(i-1)}} - \normOp{x^{(i)}}$ until such a $\rebuild(i)$ is at most $\frac{\epsilon\norm{x}}{L}$ (similarly to the decrease in the same). The proof then concludes similarly to that of \Cref{lem:sequential-size}, also using that $\normOp{\x^{(L)}}\leq \norm{x}$.
\end{proof}

    Finally, we turn to analyzing the algorithm's update time.
  
    \begin{lem}\label{lem:sequential-update}
    The (amortized) time per $\update$ of \Cref{alg:dynamic_rounding_hierarchy} is $O(\eps^{-1}\cdot L^2).$
    \end{lem}
\begin{proof}
    By \Cref{rem:Fi-ub}, 
    $|F_i|\leq 1 + 2^i\cdot \norm{x} = O(2^i\cdot \norm{x})$ (recalling that without loss of generality $\norm{x}\geq 1$). 
    Similarly, trivially $|S_i| = 2^{i}\cdot 2^{-i}\cdot |S_i|\leq 2^i\cdot \norm{x}$.
    Therefore, by \Cref{prop:degree-split}, the calls to $\dsplit(G[E_i\uplus F_{i}])$ in \Cref{alg:dynamic_rounding_hierarchy} (at which point $E_i=S_i$) take time $O(2^{i}\cdot \norm{x})$, and so the time for $\rebuild(i)$ is $\sum_{j=0}^i O(2^j\cdot \norm{x}) = O(2^i\cdot \norm{x})$.
    But since $\rebuild(i)$ is called after $2^{i-2}\cdot \frac{\eps \norm{x}}{L}$ updates, its cost amortizes to $O(\eps^{-1}\cdot L)$ time per update. 
    Summing over all $i\in \{0,1,\dots,L\}$, we find that indeed, the amortized time per $\update$ operation, which is $O(L)$ (due to deleting $O(1)$ edges from each $E_i$ and $F_i$ for each $i$) plus its contribution to periodic calls to $\rebuild$, is $O(\eps^{-1}\cdot L^2)$.
\end{proof}

    We are finally ready to prove \Cref{thm:sequential-algo}.
\begin{proof}[Proof of \Cref{thm:sequential-algo}]
    \Cref{alg:dynamic_rounding_hierarchy} is a dynamic rounding algorithm for bipartite fractional matchings, since 
    $M$ is a matching contained in $\suppOp{(\x_0)} \subseteq \suppOp{(\x_1)}\subseteq \dots \subseteq \suppOp{(\x)}$ if the latter is bipartite, by \Cref{lem:feasible-fractional} and \Cref{lem:sequential-containment}, and moreover $|M|=\normOp{\x^{(0)}}\geq (1-2\eps)\cdot \norm{x}$, by \Cref{lem:sequential-size}.
    The algorithm's $\update$ time and $\init$ time follow from \Cref{lem:sequential-update} and \Cref{thm:static_bipartite_round}.
\end{proof}

    To (nearly) conclude, this section provides a simple bipartite rounding algorithm with near-optimal $\eps$-dependence. In the following section we show how \emph{partially} rounding the fractional matching allows to dynamically guarantee that $\x_{\min}$ be sub-polynomial in $\eps/n$, thus allowing us to decrease $L$ and obtain speedups (improved $n$-dependence) when combined with \Cref{alg:dynamic_rounding_hierarchy}.

    \paragraph{\Cref{alg:dynamic_rounding_hierarchy} in general graphs.}    Before continuing to the next section, we mention that the alluded-to notion of partial rounding will also be useful when rounding (well-structured) fractional matchings in \emph{general} graphs as well (see \Cref{sec:general}).
    With this in mind, we provide the following lemma, which is useful to analyze \Cref{alg:dynamic_rounding_hierarchy} when rounding general graph matchings.

    \begin{restatable}{lem}{sequentialgeneral}\label{sequential-general}
        For $d^c_V(\x, \y) := \sum_{v \in V} (|x(v) - y(v)| - c)^+$, the vectors $\x^{(i)}$ satisfy $$d^{2^{-i+1}}_V(\x,\x^{(i)})\leq \eps\cdot \norm{x} \qquad \forall i\in \{0,1,\dots,L\}.$$
    \end{restatable}
    
\begin{proof}
        First, we verify that the inequality holds (with some extra slack) right after $\rebuild(i)$ (and in particular right after $\init$).
        First, by Property \ref{property:degree-halved} of $\dsplit$, during the invocation of which $E_i=S_i$, we have that $F_i(v) \in \left[\frac{1}{2}(E_{i+1}(v)+F_{i+1}(v) - 1,\, \frac{1}{2}(E_{i+1}(v)+F_{i+1}(v)) + 1\right]$.
        Therefore, by \Cref{rounded-frac}, for each vertex $v$, we have after $\rebuild(i)$ that $$|x^{(i)}(v)-x^{(i+1)}(v)|\leq 2^{-i}.$$
        On the other hand, 
        before an $\update$ with current input $\x$   
        there are at most $2^{i-2}\cdot \frac{\eps\norm{x}}{L}$ calls to $\update$ since the last call to $\rebuild(i)$,
        resulting in at most $3\cdot 2^{i-2}\cdot \frac{\eps\norm{x}}{L}$ many edges being added or deleted from $E_i\cup F_i$. 
        Therefore, during the updates between calls to $\rebuild(i)$, the total variation distance between $\x^{(i)}$ and $\x^{(i+1)}$ changes by at most $\frac{\eps\norm{x}}{L}$, and so after $\init$ and after any $\update$, 
        $$\sum_v \left(|x^{(i)}(v) - x^{(i+1)}(v)|-2^{-i}\right)^+\leq \frac{\eps\norm{x}}{L}.$$
        Now, using the basic fact that $(a+b)^+ \leq a^+ + b^+$ for all real $a,b$, summing the above difference over all $i$, and using that $\x=\x^{(L)}$ by \Cref{obs:xmin-and-high-bits}, we obtain the desired inequality, as follows.
        \begin{align*}
            \sum_v \left(|x^{(i)}(v) - x^{L}(v)|-2^{-i+1}\right)^+ & \leq
            \sum_{j=i}^{L-1} \sum_v \left(|x^{(i)}(v) - x^{(i+1)}(v)|-2^{-i}\right)^+ \leq L\cdot \frac{\eps\norm{x}}{L} = \eps\norm{x}. && \qedhere
        \end{align*}
\end{proof}

\section{Partial Rounding: a Path to Speedups}\label{sec:partial-rounding}

So far, we have provided a rounding algorithm with near-optimal dependence on $\eps$ (by \Cref{fact:recourse-lb}) and polylogarithmic dependence on $\x_{\min}^{-1}=\poly(\eps^{-1}n)$ of the fractional matching $\x$. 
To speed up our algorithm we thus wish to dynamically maintain a ``coarser'' fractional matching $\x'$ (i.e., with larger $(\x'_{\min})^{-1}$ than $\x_{\min}^{-1}$) that approximately preserves the value of $\x$. The following definition captures this notion of coarser fractional matchings that we will use.\footnote{In what follows, we use the definition of $d^{\epsilon}_V$ from \Cref{obs:monotone-distance}.}

\renewcommand{\bf}[1]{\mathbf{#1}}
\begin{restatable}{Def}{splitDef}\label{def:coarsening}
  Vector $\x'\in \R^E_{\geq 0}$ is an \emph{$(\epsilon, \delta)$-coarsening} of a vector $\x\in \mathbb{R}^E_{\geq 0}$ if:
\begin{enumerate}[label=(C{{\arabic*}})]  
  \setcounter{enumi}{-1}
        \item \textbf{Containment}: $\supp{x'}\subseteq \supp{x}$.  \label{property:containment}
        \item \textbf{Global Slack}: $|\norm{x} - \norm{x'}| \leq  \eps\cdot \norm{\x} + \eps$.
        \label{property:global-slack}
        \item \textbf{Vertex Slack}: 
        $d^{\epsilon}_V\left(\x, \x'\right) \leq \eps\cdot \norm{x} + \eps$.
        \label{property:vertex-wise-slack}
        \item \textbf{Edge Values}: $x'_e \in \{0\} \cup [\delta, 2\delta)$ if $x_e < \delta$ and $x'_e = x_e$ otherwise.
        \label{property:edgebound-correct}    
\end{enumerate}

The coarsening $\x'$ is \emph{bounded} if it also  satisfies the following property:

\begin{enumerate}[label=(C4)]  
        \item \textbf{Boundedness}: $x'(v) \leq x(v) + \eps$ for all $v \in V$.
        \label{property:boundedcoursening}  
\end{enumerate}
\end{restatable}

We briefly motivate the above definition: As we shall see, properties \ref{property:global-slack} and \ref{property:vertex-wise-slack} imply that $\x'$ (after mild post-processing) is a $(1-\eps)$-approximation of $\x$, and so rounding $\x'\leq \x$ results in a $(1-\eps)^2\geq (1-2\eps)$-approximation of $\x$. 
The less immediately intuitive Property \ref{property:vertex-wise-slack} will also prove useful when rounding in general graphs, in \Cref{sec:general}. 
For now, we will use this property when combining coarsenings of disjoint parts of the support of $\x$. 
Property \ref{property:edgebound-correct} then allows us to round such coarsening $\x'$ efficiently, with only a polylogarithmic dependence on $\delta^{-1}$, using \Cref{alg:dynamic_rounding_hierarchy} (by \Cref{thm:sequential-algo}). Finally, Property~\ref{property:boundedcoursening} guarantees that $\x'/(1+\eps)$ is a fractional matching.

A key ingredient for subsequent sections is thus a dynamic coarsening algorithm, as follows.

\begin{wrapper}
\begin{Def}\label{def:rounding} A \underline{\emph{dynamic $(\eps,\delta)$-coarsening algorithm}} is a data structure supporting the following operations:
\begin{itemize}
    \item $\init(G = (V,E), \, \x \in \R^E_{\geq 0})$: initializes the data structure for undirected graph $G$ with vertices $V$ and edges $E$, current vector $\x$.
    \item $\update(e \in E,\, \nu \in [0,1])$: sets $x_e \gets \nu$.
\end{itemize}
The algorithm must maintain an $(\eps,\delta)$-coarsening $\x'$ of (the current) $\x$.
\end{Def}
\end{wrapper}

As we show in \Cref{sec:general}, the internal state of \Cref{alg:dynamic_rounding_hierarchy} yields a dynamic coarsening algorithm.
In this section we state bounds for a number of dynamic coarsening algorithms (analyzed in \Cref{sec:partial-rounding-implementations}), with the objective of using their output as the input of \Cref{alg:dynamic_rounding_hierarchy}, from which we obtain faster dynamic bipartite rounding algorithms than when using the latter algorithm in isolation. The following lemma, proved in \Cref{sec:proof:claim:dynamic-coursaneing-to-rounding}, captures the benefit of this approach.

\begin{restatable}{lem}{coarsentoround}
\label{lem:dynamic-coarsening-to-rounding}
\textbf{(From coarsening to rounding).} Let $\calC$ be a dynamic $(\epsilon,\delta)$-coarsening algorithm with $\update$ time $t^{\calC}_U := t^{\calC}_U(\eps, \delta, n)$ and $\init$ time $O(|\supp{\x}| \cdot t^\calC_{I})$. Let $\calR$ be a dynamic rounding algorithm for fractional matchings $\x$ with $\x_{\min}\geq \delta$, with $\update$ time $t^{\calR}_U := t^{\calR}_U(\eps,\delta,n)$ and $\init$ time $O(|\supp{x}|\cdot t^{\calR}_I)$, for $t^{\calR}_I := t^{\calR}_I(\eps,\delta,n)$. Then, there exists  an $O(\epsilon + \delta)$-approximate dynamic rounding algorithm $\calR^*$ with $\update$ time $O(t_U^{\calC} + t^{\calR}_U + \eps^{-1}\cdot t^{\calR}_I)$ and $\init$ time $O(|\supp{\x}| \cdot (t^\calR_I + t^\calC_I))$ which is deterministic/\adaptive/\outputadaptive if both $\calR$ and $\calC$ are.
\end{restatable}

In our invocations of \Cref{lem:dynamic-coarsening-to-rounding} we will use \Cref{alg:dynamic_rounding_hierarchy} to play the role of Algorithm $\calR$. 
In \Cref{sec:partial-rounding-implementations} we provide a number of coarsening algorithm, whose properties we state in this section, together with the obtained rounding algorithms' guarantees.

A number of our coarsening algorithms will make use of subroutines for splitting (most of) the fractional matching's support into numerous disjoint coarsenings, as in the following.
\begin{Def}\label{def:split}
    An \emph{$(\epsilon, \delta)$-split} of fractional matching $\z\in \mathbb{R}_{\geq 0}^E$ with $\z_{\max} \leq \delta$ consists of $(\epsilon, \delta)$-coarsenings $\z^{(1)},\dots,\z^{(k)}$ with disjoint supports, together covering at least half of $\supp{z}$, i.e.,  $$\sum_i |\supp{z^{(i)}}|\geq \frac{1}{2}\cdot |\supp{z}|.$$
\end{Def}

The following lemma combined with \Cref{lem:dynamic-coarsening-to-rounding} motivates our interest in such splits.
\begin{restatable}{lem}{splittocoarsen}
\label{lem:split-to-dynamic-coarsening}
\textbf{(From static splitting to dynamic coarsening).}
Let $\calS$ be a static $(\gamma, \delta)$-split algorithm with running time $|\supp{x}| \cdot t_s$ on uniform fractional matching $\x$, where $t_s := t_s(n,\gamma,\delta)$. Then there exists a dynamic algorithm $\calC$  which for any (possibly non-uniform) fractional matching $\x$  maintains an $\left(O\left(\eps + \gamma \cdot \eps^{-1}\cdot (\log (\gamma^{-1}\cdot n)\right), \delta\right)$-coarsening of $\x$ with $\update$ time $O(\eps^{-1} \cdot t_s)$ and $\init$ time $O(|\supp{x}| \cdot t_s)$. Algorithm $\calC$ is deterministic/\adaptive/\outputadaptive if $\calA$ is.
\end{restatable}

\paragraph{Section outline.}
We prove \Cref{lem:dynamic-coarsening-to-rounding,lem:split-to-dynamic-coarsening} in \Cref{sec:proof:claim:dynamic-coursaneing-to-rounding,sec:proof:lem:split-to-dynamic-coarsening}, respectively. Before that, we state bounds of a number of such partial rounding algorithms (presented and analyzed in \Cref{sec:partial-rounding-implementations}), together with the rounding algorithms we obtain from these, yielding \Cref{thm:bipartite}.

\subsection{Partial Rounding Algorithms, with Applications}\label{sec:coarsening-applications}

Here we state the properties of our coarsening and splitting algorithms presented in \Cref{sec:partial-rounding-implementations}, together with the implications to dynamic rounding, as stated in \Cref{thm:bipartite}.

In \Cref{sec:proof:det-split}, we provide a deterministic static split algorithm as stated in the following lemma.

\begin{restatable}{lem}{detsplit}\label{det-split}
For any $\eps > 0$, there exists a deterministic static $(4 \eps, \eps)$-split algorithm which on input uniform fractional matchings $\x$ runs in time $O\left(|\support{x}| \cdot \log({\epsilon}^{-1} \cdot n\right))$.
\end{restatable}

Combining the above lemma with \Cref{lem:dynamic-coarsening-to-rounding,lem:split-to-dynamic-coarsening} yields the first result of \Cref{thm:bipartite}.

\begin{cor}

\label{cor:bipartite:deterministic}

There exists a deterministic dynamic bipartite rounding algorithm with $\update$ time $O(\eps^{-1} \cdot t(\eps,n))$ and $\init$ time $O(|\supp{x}|\cdot t(\eps,n))$, for $t(\eps,n) = \log n +  \log^2(\eps^{-1})$.
\end{cor}

\begin{proof}

By \Cref{det-split}, there exists a deterministic static $\left(\frac{4 \eps^3}{\log n}, \frac{\eps^3}{\log n}\right)$-split algorithm that on uniform fractional matching $\x$ runs in $O(|\supp{\x}| \cdot \log(\eps^{-1} \cdot n))$ time. Plugging this algorithm into Lemma~\ref{lem:split-to-dynamic-coarsening} yields a deterministic dynamic $\left(O(\eps), \frac{\eps^{3}}{\log n }\right)$-coarsening algorithm $\calC$ with $\update$ time $t^{\calC}_U = O(\eps^{-1} \cdot (\log n))$ and initialization time $O(|\supp{\x}| \cdot \log n)$. 
Moreover, by \Cref{thm:sequential-algo} there exists a deterministic dynamic bipartite matching rounding algorithm $\calR$ for fractional bipartite matchings $\x$ with $\x_{\min}=\Omega(\poly(\eps^{-1}\cdot \log n))$ with $\update$ time $t^{\calR}_U = O(\eps^{-1} \cdot \log^2(\eps^{-1}\cdot \log n))$ and $\init$ time $O(|\supp{\bf{x}}| \cdot t^{\calR}_I)$, for $t^{\calR}_I= O(\log(\eps^{-1}\cdot \log n))$. Plugging these algorithms into \Cref{lem:dynamic-coursening:expectation}, we obtain a deterministic algorithm which has $\update$ time $$O(t^{\calC}_U + t^{\calR}_U + t^{\calR}_I \cdot \eps^{-1}) = O(\eps^{-1} \cdot (\log n + \log^2(\eps^{-1}\cdot \log n))) = O(\eps^{-1} \cdot (\log n + \log^2(\eps^{-1})),$$ 
and initialization time $O(|\supp{\x}| \cdot (\log n + \log^2 (\eps^-1))$. The last equality holds for all ranges of $n$ and $\eps$, whether $\eps^{-1}=O(\log n)$ or $\eps^{-1}=\Omega(\log n)$.
\end{proof}

Next, in \Cref{sec:api-rand-implement} we provide a simple linear-time subsampling-based randomized split algorithm with the following properties.

\begin{restatable}{lem}{randsplitwhp}\label{rand-split-whp}
For any $\eps > 0$, there exists a static randomized algorithm 
that on uniform fractional matchings $\x$ computes an $(\eps, \frac{\eps^4}{24\log^2 n})$-split in 
$O(|\support{x}|)$-time, and succeeds w.h.p.
\end{restatable}

Combining the above lemma with \cref{lem:dynamic-coarsening-to-rounding,lem:split-to-dynamic-coarsening} yields the w.h.p.~result of \Cref{thm:bipartite}.
\begin{cor}

\label{cor:bipartite:whp}

There exists an adaptive dynamic bipartite rounding algorithm that succeeds w.h.p., with $\update$ time
$O(\eps^{-1} \cdot t(\eps,n))$ and $\init$ time $O(|\supp{x}|\cdot t(\eps,n))$, for $t(\eps,n)=\log^2 \log n+\log^2(\eps^{-1})$.
\end{cor}

\begin{proof}

By Lemma~\ref{rand-split-whp}, there exists a randomized static algorithm that computes a $\left(\frac{\eps^3}{\log(n)}, \frac{\eps^{12}}{24 \cdot \log^6(n)}\right)$-split of any uniform fractional matching $\x$ in $O(|\supp{\x}|)$ time, succeeding w.h.p. Plugging this algorithm into Lemma~\ref{lem:split-to-dynamic-coarsening} we obtain a randomized (with high probability) dynamic  $\left(O(\eps), \frac{\eps^{12}}{24 \cdot \log^6(n)}\right)$-coarsening algorithm $\calC$ with $\update$ time $t^{\calC}_U = O(\eps^{-1})$ and $\init$ time $O(|\supp{\x}|)$. 
On the other hand, by \Cref{thm:sequential-algo}, there exists a deterministic dynamic bipartite matching rounding algorithm $\calR$ for fractional matchings $\x$ with $\x_{\min}=\Omega(\poly(\eps^{-1}\cdot \log n))$ with $\update$ time $t^{\calR}_U = O(\eps^{-1} \cdot \log^2(\log n \cdot \eps^{-1}))$ and $\init$ time $O(|\supp{\bf{x}}| \cdot t^{\calR}_I)$, for $t^{\calR}_I= O(\log(\log n \cdot \eps^{-1}))$. Plugging these algorithms into \Cref{lem:dynamic-coursening:expectation}, we obtain a randomized \adaptive algorithm which works with high probability and has $\update$ time $$O(t^{\calC}_U + t^{\calR}_U + t^{\calR}_I \cdot \eps^{-1}) = O(\eps^{-1} \cdot \log^2(\log n \cdot \eps^{-1})) = O\left(\eps^{-1}\cdot \left(\log^2\log n + \log^2(\eps^{-1})\right)\right),$$
and $\init$ time $O(|\supp{\x}| \cdot \log^2(\log n) + \log^2 (\eps^{-1}))$. The last equality holds whether $\eps^{-1}=O(\log n)$ or $\eps^{-1}=\Omega(\log n)$.
\end{proof}

Finally, in \Cref{sec:proof:lem:dynamic-coursening:expectation}, building on a \outputadaptive dynamic set sampling algorithm which we provide in \Cref{sec:app:setsampling},
we give an \outputadaptive coarsening algorithm with constant (and in particular independent of $n$) expected amortized $\update$ time.

\begin{restatable}{lem}{randsplitexp}
    
\label{lem:dynamic-coursening:expectation}

There exists an \outputadaptive dynamic 
$(O(\eps), \eps^{3})$-coarsening algorithm for dynamic fractional matchings $\x$
with expected $\update$ time $O(\eps^{-1})$ and expected $\init$ time $O(|\supp{\x}|)$.
\end{restatable}

Finally, combining the above lemma with \Cref{lem:dynamic-coarsening-to-rounding}  yields the third result of \Cref{thm:bipartite}.

\begin{cor}

\label{cor:bipartite:expectation}

There exists an \outputadaptive dynamic bipartite rounding algorithm with expected $\update$ time $O(\eps^{-1} \cdot t(\eps))$ and expected $\init$ time $O(|\supp{\x}| \cdot t(\eps))$ for $t(\eps) = \log^2(\eps^{-1})$.
\end{cor}

\begin{proof}
By \Cref{lem:dynamic-coursening:expectation}, there exists a dynamic $(\eps,O(\eps^3))$-coarsening algorithm $\calC$ with expected $\update$ time $t^{\calC}_U=O(1)$ and $\init$ time $O(|\supp{\x}|)$. 
On the other hand, by \Cref{thm:sequential-algo} there exists a deterministic (hence \outputadaptive) dynamic bipartite matching rounding algorithm $\calR$ 
for fractional matchings $\x$ with $\x_{\min}=\Omega(\eps^{3})$ with $\update$ time $t^{\calR}_U = O(\eps^{-1} \cdot \log^2(\eps^{-1}))$ and $\init$ time $O(|\supp{\bf{x}}| \cdot t^{\calR}_I)$, for $t^{\calR}_I= O(\log(\eps^{-1}))$.
Plugging these algorithms into \Cref{lem:dynamic-coursening:expectation}, we obtain an \outputadaptive algorithm with expected $\update$ time $O(t^{\calC}_U + t^{\calR}_U + t^{\calR}_I \cdot \eps^{-1}) = O(\eps^{-1} \cdot \log^2(\eps^{-1}))$ and expected $\init$ time $O(|\supp{\x}| \cdot \log^{2}(\eps^{-1}))$.
\end{proof}
\subsection{Proof of Lemma~\ref{lem:dynamic-coarsening-to-rounding}: Reducing Rounding to Coarsening}

\label{sec:proof:claim:dynamic-coursaneing-to-rounding}

The following lemma allows us to efficiently convert coarsenings to bounded coarsenings.

\begin{lem}

\label{cl:bounded-coarsening}

There exists a deterministic algorithm which given an $(\eps,\delta)$-coarsening $\x'$ of fractional matching $\x$, finds in $O(|\supp{\x'}|)$ time a bounded $(3 (\eps + \delta), \delta)$-coarsening $\x''$ of $\x$, with $\x''_e=0$ only if $x_e < \delta$.
\end{lem}

\begin{proof}

\textbf{The algorithm:} Initialize $\x'' \leftarrow \x'$. For any vertex $v \in V$ such that $x''(v) > x(v) + \epsilon + 2\delta,$ remove arbitrary edges $e$ incident on $v$ in $\supp{\x''}$ such that $x_e \leq \delta$ until $x''(v) \leq x(v) + \epsilon + 2\delta$. Note that edges $e$ with $x_e\geq \delta$ have $x'_e=x_e$, and so the above process must terminate as once all edges with weight at most $\delta$ are removed for all vertices $v \in V$ we have $x_v'' \leq x_v$. Finally, return $\x''$.

\textbf{Running time:} Each edge of $\supp{\x''}$ has its value decreased  at most once, hence the running time of the algorithm is at most $O(|\supp{\x'}|)$. 

\textbf{Correctness:}
By construction and by Property \ref{property:containment} of coarsening $\x'$ of $\x$, we have that $\supp{\x''} \subseteq \supp{\x'} \subseteq \supp{\x}$, and so $\x''$ satisfies Property \ref{property:containment} of a coarsening of $\x$.
Next, since $\x'$ is an $(\eps,\delta)$-coarsening of $\x'$ and $\x''$ agrees with $\x'$ on $\supp{x''}$, we find that $\x''$ also satisfies Property \ref{property:edgebound-correct}.
Moreover, by definition, the algorithm ensures that $x''(v) \leq x(v) + \epsilon + 2\delta$ for any vertex $v \in V$, and so $\x'$ satisfies Property~\ref{property:boundedcoursening} of bounded $(3(\eps+\delta),\delta)$-coarsenings of $\x$. To prove the remaining  properties we leverage some minor calculations, which we now turn to.

Consider some vertex $v$ that had one of its edges in $\supp{\x''}$ deleted. Since before the deletion it must be that $x''(v) > x(v) + \epsilon + 2\delta$ and the deleted edge had weight at most $\delta$, we must have that after the deletion $x''(v) \geq x(v) + \epsilon$. This implies that the total weight of deleted edges is upper bounded by $d^\epsilon_V(\x, \x') \leq \norm{\x} \cdot \eps + \eps$, as $\x'$ is an $(\eps, \delta)$-coarsening of $\x$. Hence, $|\norm{\x'} - \norm{\x''}| \leq \norm{\x} \cdot \eps + \eps$. Note that $|\norm{\x} - \norm{\x'}| \leq \norm{\x} \cdot \eps$ as $\x'$ is an $(\eps, \delta)$-coarsening of $\x$. Thus, $\x''$ satisfies Property~\ref{property:global-slack} of $(3(\eps + \delta), \delta)$-coarsenings of $\x$:
$$ |\norm{\x''} - \norm{\x}| \leq |\norm{\x'} - \norm{\x}| + |\norm{\x'} - \norm{\x''}| \leq \norm{\x} \cdot 2 \cdot \eps + 2 \eps.$$

By Property \ref{property:vertex-wise-slack}, $d_V^{\eps}(\x, \x') \leq \norm{\x} \cdot \eps + \eps$. Therefore, by \Cref{obs:monotone-distance} and the above, we find that $\x''$ also satisfies Property~\ref{property:vertex-wise-slack} of $(3(\eps + \delta), \delta)$-coarsenings of $\x$. 
\begin{align*}
d_V^{3(\eps + \delta)}(\x, \x'') & \leq d_V^\eps(\x, \x') + d_V^{0}(\x',\x'') \leq d_V^\eps(\x, \x') + 2 \cdot |\norm{\x''} - \norm{\x'}| \leq \norm{\x} \cdot 3 \cdot \eps  +  3\eps. \qedhere 
\end{align*}
\end{proof}

With Claim~\ref{cl:bounded-coarsening} established, we are now ready to prove \cref{lem:dynamic-coarsening-to-rounding}.

\begin{proof}[Proof of \cref{lem:dynamic-coarsening-to-rounding}]
We will describe the algorithm $\calR^*$. The input dynamic fractional matching of $\calR^*$ is denoted by $\x$. Recall that $\bf{x_{\leq 2\delta}}$ and $\bf{x_{> 2\delta}}$ refer to $\x$ restricted to edges with weight at most $2\delta$ and greater then $2\delta$ respectively. Algorithm $\calR^*$ uses $\calC$ to always maintain fractional matching $\x^{A^*}$ an $(\eps, 2\delta)$-coarsening of $\x$. Define scaling constant $\alpha = (1 + 3(\eps + 2\delta))^{-1}$, and assume $\alpha \geq 1/2$. In addition $\calR^*$ completes the following operations:

\textbf{Initialization:} $\calR^*$ uses the algorithm of \Cref{cl:bounded-coarsening} to obtain a bounded $(3\cdot(\eps+2\delta), 2\delta)$-coarsening $\x^{A^*}_N$ of $\bf{x}$ from $\x^{A^*}$. $\calR^*$ sets $\x_{Small} \leftarrow \x^{A^*}_N$ restricted to edges of $\supp{x_{\leq 2\delta}}$ and $\x_{Large} \leftarrow \x^{A^*}_N$ restricted to edges of $\suppOp(\x_{>2\delta})$. The algorithm initializes its output $\x'$ by providing $\hat{\x} = (\x_{Small} + \x_{Large})\cdot \alpha$ as input to $\calR$. We will maintain throughout the algorithm, that $\hat{\x} = (\x_{Small} + \x_{Large})\cdot \alpha$ and that $\x'$ is the output of $\calR$ when given input fractional matching $\hat{\x}$. The algorithm furthermore sets counter $C \leftarrow 0$. Let $\bf{\hat{x}}_0$ stand for the state of $\bf{\hat{x}}$ at initialization.

\textbf{Handling an edge update:} For the sake of simplicity, we assume every update either changes the weight of some edge $e=(u,v)$ from $0$ to some positive value or vice versa, thus inserting or deleting an edge to/from $\supp{\x}$.
If $x_e > 2\delta$, algorithm $\calR^*$ removes $e$ (if present) from $\suppOp(\x_{Large})$ or adds $e$ to $\suppOp(\x_{Large})$ with weight $x_e$ and then updates $\x'$ using $\calR$. Note that by assumption $\alpha \geq 1/2$, hence if $e$ was inserted into $\x$ then $\bf{\hat{x}}$ undergoes an edge update with weight at least $2\delta \cdot \alpha \geq \delta$ and hence $\calR$ can handle this update correctly.

If $x_e \leq 2 \delta$, the algorithm removes one arbitrary edge incident on $u$ and $v$ from $\suppOp(x_{Small})$ (if there is any) and administer the changes this makes on $\bf{\hat{x}}$ and $\x'$ using $\calR$. Furthermore, if $e$ was deleted from $\supp{\x}$, then remove $e$ from $\suppOp(\x_{Small})$ and update $\x'$ with $\calR$. If $e$ was inserted then its effect (apart from the updates to its endpoints) is ignored.
Either way, set $C \leftarrow C + 12 \cdot \delta$ and re-initialize the datastrucrure if $C > \norm{\x_0} \cdot \eps$.

\textbf{Running Time:} First note that at initialization the algorithm first needs to initialize $\x^{\calA^*}$ on $\x_{\leq 2\delta}$ which takes $O(t_I^{\calC} \cdot |\supp{\x}|)$ time. Afterwards $\calR^*$ constructs fractional matching $\x^{\calA^*}_N$ in $O(|\supp{\x}|)$ time. Finally, the algorithm needs to run $\calR$ on the fractional matching $\hat{\x}$ which takes $O(t_I^{\calR} \cdot |\supp{\hat{\x}}|)$ time. $|\supp{\hat{\x}}| = |\suppOp(\x_{Small})| + |\suppOp(\x_{Large})|$ where $\suppOp(\x_{Large}) = \suppOp(\x_{ > 2\delta})$ and $\suppOp(\x_{Small}) \subseteq \suppOp(\x_{ \leq 2 \delta})$. Hence, initialization takes time $O(|\supp{\x}| \cdot (t_I^\calC+ t_I^\calR))$. 

The fact that the algorithm maintains an integral matching $M\subseteq \supp{\x}$ of expected size at least $|M|\geq \norm{\x} \cdot (1-O(\epsilon+\delta))$ follows from Claim~\ref{cl:correctness:lem:dynamic-coarsening-to-rounding} and the properties of $\epsilon$-rounding algorithms. Observe that unless the algorithm re-initializes after an update, it only uses $\calR$ to handle $O(1)$ updates on the support of $\hat{\x}$, which takes update time $O(t^{\calR}_U)$. The algorithm has to run $\calC$ on the dynamic fractional matching $\x$ at all times which adds an update time of $t^{\calC}_U$. Furthermore, the algorithm sometimes re-initializes.

Consider the cost of initialization at time $0$, which took $O(|\supp{\hat{\x}_0}| \cdot t^{\calR}_U)$ time. As $\hat{\x}$ has edge weights at least $2\delta \cdot \alpha\geq\delta$ we know that $|\supp{\hat{\x}_0}| = O(\norm{\bf{\hat{x}}_0} / \delta)$. By the next initialization $C>\norm{\x_0} \cdot \eps$ hence at least $\Omega(\norm{\x_0} \cdot \eps /\delta)$ updates must have occurred. By Claim~\ref{cl:correctness:lem:dynamic-coarsening-to-rounding} we have that $\norm{\x} = \Theta(\norm{\hat{\x}})$. Amortizing the cost of the initialization over these updates yields an additional expected update time cost of $O(t^{\calR}_U \cdot \eps^{-1})$. Hence, the total expected amortized update time of the algorithm is $O(t^{\calC}_U + t^{\calR}_U + t^{\calR}_I \cdot \eps^{-1})$. 

\textbf{Adaptivity:} Note that other than the inner operations of $\calC$ and $\calR$, the algorithm's actions are deterministic, and hence it is deterministic/\adaptive/\outputadaptive if both $\calC$ and $\calR$ are.

\begin{claim}

\label{cl:correctness:lem:dynamic-coarsening-to-rounding}

At all times $\norm{\hat{\x}} \geq \norm{\x} \cdot (1 - O(\epsilon+\delta))$ and $\hat{\x}$ is a valid fractional matching.

\end{claim}

\begin{proof}

We will first argue that $\hat{\x}$ remains a valid fractional throughout the set of updates. By definition throughout the run of the algorithm $\hat{\x} = (\x_{Small} + \x_{Large}) \cdot \alpha$. Fix some vertex $v \in V$. As $x_{Large}(v) = x_{> 2\delta}(v)$ and $\x$ is a valid fractional matching it is sufficient to argue that $x_{Small}(v) \leq x_{\leq 2\delta}(v) + 3 \cdot (\eps + 2\delta)$ at all times. Note that this holds at initialization due to Property~\ref{property:boundedcoursening} of bounded coarsenings. Assume an edge update occurs to $\x_{\leq 2\delta}$ to edge $(u,v)$. As all edges of $\x_{\leq 2\delta}$ have weight at most $2\delta$ this update may decrease $x_{Small}(u)$ and $x_{Small}(u)$ by at most $2\delta$. The algorithm compensates for this through deleting an edge incident on both $u$ and $v$ if any are present in $\suppOp(\x_{Small})$ (which all have weight at least $2\delta$ in $\x_{Small}$). Hence for any edge $e$ in $\suppOp(\x_{Small})$ we have that $x_{Small}(e) \geq 2\delta$ after these deletions $x_{Small}(u) \leq x_{\leq 2\delta}(u) + 3 \cdot (\eps + 2\delta)$ and $x_{Small}(v) \leq x_{\leq 2\delta}(v) + 3 \cdot (\eps + 2\delta)$ hence $\bf{\hat{x}}$ remains a valid fractional matching.

It remains to argue that $\norm{\hat{\x}} \geq \norm{\x} \cdot (1 - O(\epsilon+\delta))$. Observe, that at initialization the inequality holds (as $\x_{Small}$ is a $(3(\epsilon + 2\delta), \delta)$-coarsening of $\x_{\leq 2\delta}$) and $\norm{\x_{> 2\delta}} \leq \norm{\x_{Large}} / \alpha$. As edges of $\x_{Small}$ are in $[2\delta, 4\delta]$ we also must have that by definition $(\norm{\x_{Small}} + C) \cdot \alpha \geq \norm{\x_{\leq 2\delta}}$ throughout the run of the algorithm as $C$ is increased by $4 \delta$ whenever an edge update occurs to $\hat{\x}$. As whenever $C > \norm{\hat{\x}_0} \cdot \eps$ we re-initialize we must have that $\norm{\hat{\x}} \geq \norm{\x} \cdot (1 - O(\epsilon + \delta))$ at all times.
\end{proof}

This concludes the proof of \cref{lem:dynamic-coarsening-to-rounding}.
\end{proof}
\subsection{Proof of Lemma~\ref{lem:split-to-dynamic-coarsening}: Reducing Dynamic Coarsening to Static Splitting}

\label{sec:proof:lem:split-to-dynamic-coarsening}

We start with a lemma concerning the ``stability'' of coarsenings under (few) updates.

\begin{lem}

\label{lem:coarsening:stability}

Let $\x^{(0)}$ and $\x^{(t)}$ be two fractional matchings with maximum edge weights $\delta$ which differ on at most $\norm{\x^{(0)}} \cdot \eps \cdot \delta^{-1}$ many edges. Assume that $\x^{(0)'}$ is a $(\gamma, \delta)$-coarsening of $\x^{(0)}$. Define $\x^{(t)'}$ to be $\x'$ restricted to edges of $\supp{\x^{(t)}}$. Then $\x^{(t)'}$ is a $(30\eps + 2\gamma, \delta)$-coarsening of $\x^{(t)}$.

\end{lem}

\begin{proof}
Property~\ref{property:containment} of coarsenings follows trivially as $\x^{(t)'}$ is restricted to $\supp{\x^{(0)}}$ by definition. As $\x'$ is an $(\gamma, \delta)$-coarsening of $\x^{(0)}$ we are guaranteed by Property~\ref{property:edgebound-correct} that edges of $\x'$ take weight in $[\delta, 2\delta)$, implying Property~\ref{property:edgebound-correct}.

First note that as $\x^{(0)}$ and $\x^{(t)}$ differ on at most $\norm{\x^{(0)}} \cdot \delta^{-1} \cdot \eps$ edges. Each edge they differ on takes weight at most $\delta$. This first implies that $\norm{\x^{(0)}} \leq \norm{\x^{(t)}}(1+2\eps)$. Furthermore, summing the difference over all vertices implies that $d_V^0(\x^{(0)}, \x^{(0)}) \leq \norm{\x^{(0)}} \cdot 6 \cdot \eps$.

Similarly, $\x^{(0)'}$ and $\x^{(t)'}$ may only differ on at most $\norm{\x^{(0)}} \cdot \delta^{-1} \cdot \eps$ edges. On each of these edges they both take values in $[\delta, 2\delta)$. Hence, $d_V^0(\norm{\x^{(0)}}, \norm{\x^{(t)}}) \leq \norm{\x^{(0)}} \cdot 4 \eps$.

To conclude Property~\ref{property:global-slack} of coarsenings consider the following line of inequalities:
\begin{align}
|\norm{\x^{(t)}} - \norm{\x^{(t)'}}| & \leq 2 \cdot d_V^0(\x^{(0)}, \x^{(t)}) + | \norm{\x^{(0)}} - \norm{\x^{(0)'}}| + 2 \cdot d_V^0(\x^{(0)'}, \x^{(0)'}) \label{eq:1:lem:coarsening:stability} \\
& \leq \norm{x^{(0)}} \cdot (\gamma + \eps \cdot (12 + 8)) + \gamma \nonumber \\
& \leq \norm{x^{(t)}} \cdot (\gamma \cdot 2 + \eps \cdot 30) + \gamma. \label{eq:2:lem:coarsening:stability} 
\end{align}
Inequality~\ref{eq:1:lem:coarsening:stability} follows from the definitions of the distance $d^{\eps}_V$ and of norms. Inequality~\ref{eq:2:lem:coarsening:stability} follows from the observation that $\norm{\x^{(0)}} \leq \norm{\x^{(t)}}(1+2\eps)$.

Property~\ref{property:vertex-wise-slack} follows from a similar set of inequalities:
\begin{align*}
d_V^{30\gamma}(\x^{(t)}, \x^{(t)'}) & \leq d_V^0(\x^{(0)}, \x^{(t)}) + d_V^{\gamma}(\x^{(0)},\x^{(0)'}) + d_V^0(\x^{(0)'}, \x^{(0)'}) \nonumber \\
& \leq \norm{x^{(0)}} \cdot (\gamma  + \eps \cdot (6 + 4)) + \gamma \nonumber \\
& \leq \norm{x^{(t)}} \cdot (2 \gamma + 30 \eps) + \gamma. \qedhere
\end{align*}
\end{proof}

\begin{cor}

\label{cor:coarsening:stability:uniform}

The statement of \Cref{lem:coarsening:stability} holds assuming $|\supp{\x^{(0)}}| \cdot \gamma$ and $|\supp{\x^{(t)}}| \cdot \gamma$ are uniform fractional matchings (of the same uniform value) and they differ on at most $|\supp{\x^{(0)}}| \cdot \gamma$ edges while $\x^{(0)'}$ and $\x^{(t)'}$ differ on at most $|\supp{\x^{(t)'}}|$ edges. Furthermore, $\x^{(t)'}$ satisfies slightly stronger slack properties, $\left|\norm{\x} - \norm{\x'}\right| \leq \norm{\bf{x}} \cdot (2\gamma + 30\eps) + \gamma$ and $d_V^{\gamma}(\x, \x') \leq \norm{\x} \cdot (2\gamma + 30\eps) + \gamma$.

\end{cor}

To observe Corollary~\ref{cor:coarsening:stability:uniform} note that all the inequalities in the proof of Lemma~\ref{lem:coarsening:stability} hold in this slightly modified setting.
In Claim~\ref{cl:dynamic-coarsening:uniform:rough}, we   show that if the input fractional matching $\x$ is uniform, then we can maintain a coarsening of $\x$ efficiently as it undergoes updates. Afterwards we show how to extend the argument to general fractional matchings.

\begin{lemma}

\label{cl:dynamic-coarsening:uniform:rough} 

Let $\calS$ be a static $(\gamma, \delta)$-split algorithm that on uniform fractional matching $\x$ takes time $O(|\supp{x}| \cdot t_s)$, for $t_s = t_s(n,\gamma,\delta)$.
Then there exists a dynamic $(\epsilon+\gamma, \delta)$-coarsening algorithm $\calC_U$ for uniform fractional matchings, whose output 
$(\epsilon+\gamma, \delta)$-coarsening $\bf{x'}$ of $\bf{x}$ satisfies slightly stronger slack properties, namely:
\begin{enumerate}
[label=(C{{\arabic*}}')]
        \item \textbf{Stronger Global Slack}: $\left|\norm{\x} - \norm{\x'}\right| \leq \norm{\bf{x}} \cdot (\eps + \gamma) + \gamma$. \label{property:global-slack:stronger}  
        \item \textbf{Stronger Vertex Slack}: $d_V^{\gamma/4}(\x, \x') \leq \norm{\x} \cdot (\eps + \gamma) + \gamma$.
        \label{property:vertex-wise-slack:stronger}  
\end{enumerate}
 Algorithm $\calC_U$ has $\update$ time $O(\eps^{-1}\cdot t_s)$ and $\init$ time $O(|\supp{\bf{x}}|\cdot t_s)$, and is deterministic/adaptive if $\calS$ is deterministic/randomized.
\end{lemma}

\begin{proof}
The algorithm $\calC_U$ works as follows.

\textbf{Initialization:} $\calC_U$ calls $\calS$ to compute a $(\gamma/4,\delta)$-split 
$\x_1, \x_2, \dots$ of $\x$. Next, $\calC_U$ sets its output $\x'$ to be $\x_1$. Let $\x^{(0)}$ and $\x^{(0)}_i$ respectively denote  the states of $\x$ and $\x^i$, at initialization. Note that this implies an $\init$ time of $O(|\supp{\x}| \cdot t_s)$.

\textbf{Handling an Update:} If an edge $e$ gets deleted from $\supp{\bf{x}}$, then $\calC_U$ removes $e$ from $\supp{\x'}$ as well (provided we had $e \in \supp{\x'}$). Once $|\supp{\x^{(0)}}| \cdot \eps/64$ updates have occurred to $\bf{x}$ since the last initialization,  $\calC_U$ re-initializes. Further, once more than $|\supp{\x'}| \cdot \eps/32$ edges have been deleted from $\supp{\bf{x'}}$,  $\calC_U$ discards $\bf{x'}$ from memory and switches its output to be another coarsening of the split which satisfies that at most an $\eps/8$ fraction of its support has been deleted so far. The effects of insertions are ignored, except that they contribute to the counter timing the next re-initialization.

\textbf{Update Time:} Observe that the algorithm re-initializes every $\Omega(\eps \cdot |\supp{\x^{(0)}}|)$ updates. Hence, the re-initializations have an amortized update time of $O(t_s \cdot \eps^{-1})$. If at some point in time, for some coarsening $\x^i$ at least $|\supp{\x^{(0)}_i}| \cdot \eps/32$ edges have been deleted from $\supp{\x^{(0)}_i}$, then $\x^i$ can be discarded from memory as it will never again enter the output in future. Hence, $\calC_U$ adds/removes each edge of the initially computed split exactly once  to/from the output. Accordingly, the total work of $\calC_U$ between two re-initializations is proportional to $O(|\cup_i  \supp{\x^{(0)}_i}|)$, which is at most $O(|\supp{\bf{x^{(0)}}}|)$ by definition. We  similarly amortize this cost over $\Omega(\eps \cdot |\supp{\bf{x^{(0)}}}|)$ updates, which leads  to  an update time of $O(\eps^{-1})$.

Observe that at all times $\x'$ is a $(\gamma/4, \delta)$-coarsening of $\x$-s state at the last re-initialization restricted to edges which have not been deleted since the last re-initialization. This implies both Properties \ref{property:containment} and \ref{property:edgebound-correct}.

\textbf{Correctness:}
Let $\x^{(t)}$ denote for the state of $\x$ at some time  $t$ after initialization (but before the next re-initialization). $\x^{(t)}$ and $\x^{(0)}$ may differ on at most $|\supp{\x^{(0)}}| \cdot \eps/64$ edges. Suppose that  the coarsening $\x'$ is the  output of the algorithm at time $t$, and let $\x^{(0)'}$ and $\x^{(t)'}$ respectively denote the states of the coarsening $\bf{x'}$ at initialization and at time $t$. Since the algorithm would have discarded $\x'$ from memory if more than $|\supp{\x^{(0)'}}| \cdot \eps/32$ edges of $\supp{\x'}$ got deleted we know that $\x^{(0)'}$ and $\x^{(t)'}$ differ on at most $|\supp{\x^{(0)'}}| \cdot \eps/32$ edges. Based on this the correctness of the output of the algorithm follows from \Cref{cor:coarsening:stability:uniform}.

It now remains to argue that $\calC_U$ always maintains an output, i.e., the algorithm doesn't discard all coarsenings present in the initial split before re-initialization. By definition of splits, we know that $|\bigcup_i \supp{\x^{(0)}_i}| \geq |\supp{\x^{(0)}}| / 2$, and the supports of $\{\x^{(0)}_i\}$ are disjoint. The algorithm re-initializes after $|\supp{\x^{(0)}}| \cdot \eps / 64$ updates. Hence, less then a $\eps/32$ fraction of $\bigcup_i \supp{\x^{(0)}_i}$ gets deleted between successive re-initializations. Thus,  there is always a coarsening in the split that has less then $\eps/32$ fraction deleted edges (due to the pigeonhole principle), and so $\calC_U$ always maintains a correct output throughout.
\end{proof}

\paragraph{Proof of Lemma~\ref{lem:split-to-dynamic-coarsening}:} We will now show how to extend Claim~\ref{cl:dynamic-coarsening:uniform:rough} to non-uniform fractional matching and thus prove Lemma~\ref{lem:split-to-dynamic-coarsening}. Assume that $\bf{x}$ is a fractional matching. Let $\phi = \frac{\gamma}{n^2}$. For the sake of convenience assume that $\delta = \phi \cdot (1+\epsilon)^L$ for some integer $L < 2 (\log n + \log(\gamma^{-1}))\cdot \eps^{-1}$. Define $\bf{x}_i = \bf{x}^{\geq \phi \cdot (1+\eps)^{L-1}}_{<\phi \cdot (1+\eps)^{L}}$. For $i \in [L]$ define fractional matching $\hat{\bf{x}}_i$ to be a uniform fractional matching with weight $\phi \cdot (1+\epsilon)^{L-1}$ and support $\bf{x}_i$. Define $\hat{x} = \sum_{i \in [L]} \hat{\bf{x}}_i$. First note that as the graph has at most $n^2$ edges $\norm{\bf{x}_{<\phi}} \leq \gamma$. Furthermore, for any edge $e \in \supp{\bf{x}_{\geq \phi}}$ we must have that $x_e \leq \hat{x}_e \cdot (1 + \eps)$. Hence, $|\norm{\bf{\hat{x}}} - \norm{\bf{x}}| \leq \eps \cdot \norm{\bf{x}}  + \gamma$.

Algorithm $\calC$ will use algorithm $\calC_U$ from Claim~\ref{cl:dynamic-coarsening:uniform:rough} to maintain $(\eps + \gamma, \delta)$-coarsenings $\x_i'$ of fractional matchings $\hat{\x}_i$ in parallel. The output of $\calC$ will be $\x' = \sum_{i \in L} \x_i' + \x_{> \delta}$. In order to conclude the amortized update time bound observe that an edge update to $G$ may only effect a single fractional matching $\hat{\x}_i$. Note that at all times $\supp{\x'}$ is a subset of the support of one of the initially calculated coarsenings restricted to non-deleted edges, hence $\supp{\x'} \subseteq \supp{\x}$. Note that $\x'$ and $\x$ agrees on edges $e$ where $\x(e) > \delta$.
We now turn to proving the remaining properties of a coarsening:

\medskip
\noindent
\textbf{Property~\ref{property:global-slack} Global Slack}: By Claim~\ref{cl:dynamic-coarsening:uniform:rough}, we must have that $|\norm{\hat{\x}_i}- \norm{\x_i'}| \leq \norm{\hat{\x}_i} \cdot (\eps + \gamma) + \gamma$. 
\begin{align*}
\norm{\bf{x}}- \norm{\bf{x}'}|& \leq \norm{x_{<\phi}} + \sum_{i \in [L]} |\norm{\hat{\bf{x}}_i}- \norm{\bf{x}'}| + |\norm{\hat{\bf{x}}_i}- \norm{\bf{x}_i}| \nonumber \\
& \leq \gamma + \sum_{i \in [L]} \norm{\hat{\bf{x}}_i} \cdot (\eps + \gamma) + \gamma +  \norm{\bf{x}_i} \cdot \eps\nonumber \\
& \leq \norm{\bf{x}} \cdot 4(\eps + \gamma) + \gamma \cdot 2L \nonumber \\
& \leq \norm{\bf{x}} \cdot 8(\eps + \gamma) + \gamma \cdot 4\log(n/\gamma) \cdot \eps^{-1}.\nonumber
\end{align*}

\medskip
\noindent
\textbf{Property~\ref{property:vertex-wise-slack} Vertex Slack}: $ d_V^{0}(\bf{x}, \bf{\hat{x}}) \leq 2 \cdot \gamma + 2\eps \cdot \norm{\bf{x}}$ as the two matching differ only on the edges of $\bf{x_{< \phi}}$ or by just a factor of $\eps$.
\begin{align*}
    d_V^{\eps + \gamma}(\bf{x}, \bf{x'})& \leq d_V^{0}(\bf{x}, \bf{\hat{x}}) + d_V^{\eps + \gamma}(\bf{\hat{x}}, \bf{x}')\nonumber \\
   & \leq 2 \cdot \gamma + 2\eps \cdot \norm{\bf{x}} + \sum_{i \in [L]} d_V^{\eps +  \gamma}(\bf{\hat{x}_i}, \bf{x}_i') \nonumber \\
   & \leq 2 \cdot \gamma + 2\eps \cdot \norm{\bf{x}} + \sum_{i \in [L]} 2\cdot \norm{\bf{\hat{x}_i}} \cdot (\gamma + \eps) + \gamma\nonumber \\
   & \leq \norm{\bf{x}} \cdot 4 \cdot (\eps + \gamma)  + \gamma \cdot \frac{5 \cdot \log(n/\gamma)}{\eps}.\nonumber 
\end{align*}

\medskip
\noindent
\textbf{Property~\ref{property:edgebound-correct} Edge values}: The property follows as the output is the sum of edge-wise disjoint $(\eps + \gamma, \delta)$-coarsenings.

\section{Coarsening and Splitting Algorithms}\label{sec:partial-rounding-implementations}
So far, we have provided \Cref{lem:dynamic-coarsening-to-rounding,lem:split-to-dynamic-coarsening} which give a reduction from (faster) dynamic rounding to dynamic coarsening, and from dynamic coarsening to static splitting.
We further stated a number of such dynamic coarsening and static splitting algorithms, as well as their corollaries for faster rounding algorithms.
In this section we substantiate and analyze these stated coarsening and splitting algorithms.

\subsection{Deterministic Static Splitting}

\label{sec:proof:det-split}

In this section we prove \Cref{det-split}, restated below for ease of reference.

\detsplit*

Assume $\x$ is $\lambda$-uniform. First note that if $\lambda > \eps$ then we may return $\{\x\}$ as the split trivially. On the other extreme end, if $\lambda \leq \eps^2 \cdot n^{-2}$ then we may return a split consisting of $|\supp{\x}|$ many $\eps$-uniform fractional matchings each having the support of a single edge in $\supp{\x}$ (the properties of coarsening follow trivially). Hence, we may assume that $\eps^2 \cdot n^{-2} \leq \lambda \leq \eps$. Let $\epsilon' \in [\epsilon,2\epsilon]$ be some constant such that $\lambda \cdot 2^L = \epsilon'$ for an integer $L= O\left(\eps^{-1} \cdot \log n \right)$.

\textbf{The algorithm:}
Our algorithm inductively constructs sets of vectors $\calF_i$ for $i \in \{0,1,\dots,L\}$. As our base case, we let $\calF_0 := \{\x\}$. Next, for any vector $\x' \in \calF_{i-1}$, using $\dsplit$ (See~\Cref{prop:degree-split}) on $\supp{\x'}$, we compute two edge sets $E^1_{\x'}, E^2_{\x'}$, and add to $\calF_{i}$ a $\lambda \cdot 2^i$-uniform vector on each of these two edge sets. 
The algorithm outputs $\calF_L$. 

\textbf{Running time:} 
By \Cref{prop:degree-split}, each edge in $\supp{x}$ belongs to the support of exactly one vector in each $\calF_i$ and so each $\calF_i$ is found in time $O(|\supp{x}|)$ based on $\calF_{i-1}$. Therefore, the algorithm runs in the claimed $O(|\supp{\x}| \cdot L) = O(|\supp{\x}| \cdot \log(\eps^{-1} \log n)$ time. 

\textbf{Correctness:}
It remains to show that $\calF_L$ is a $(4\epsilon, \epsilon)$-split of $\x$.
First, as noted above, each edge in $\supp{x}$ belongs to the support of precisely one vector in $\calF_L$, and so these vectors individually satisfy Property \ref{property:containment} of coarsenings, and together they satisfy the covering property of splits. We now show that every vector $\x_L\in \calF_L$ satisfies the remaining properties of $(4\epsilon,\epsilon)$-coarsening of $\x$. For $i \in \{1,\dots, L-1\}$, inductively define $\x^{i-1}$ to be the fractional matching in $\calF_{i-1}$ that was split using $\dsplit$  to generate $\x^{i}$ (hence $\bf{x_0} = \x$). By Property \ref{property:size-halved} of \Cref{prop:degree-split}: 
$$|\suppOp(\x^i)| \in \left[\frac{|\suppOp(\x^{i-1})|}{2} - 1, \frac{|\suppOp(\x^{i-1})|}{2} + 1\right] \qquad \text{ for } i \in \{1,2,\dots,L\}.$$ Thus, as $\x^i$ and $\x^{i-1}$ are respectively $\lambda\cdot2^i$ and $\lambda \cdot 2^{i-1}$-uniform vectors, $|\normOp{\x^i} - \normOp{\x^{i-1}}| \leq \lambda \cdot 2^i$. Therefore, Property \ref{property:global-slack} follows form the triangle inequality, as follows.
    $$|\norm{\x} - \normOp{\x^L}| = |\normOp{\x^0} - \normOp{\x^L}| \leq \sum_{i \in [L]} \left|\normOp{\x^i} - \normOp{\x^{i-1}}\right| 
     \leq  \sum_{i \in [L]} \lambda \cdot 2^i 
     \leq 2\eps'\leq 4 \epsilon.$$

Similarly, this time by Property \ref{property:degree-halved} of \Cref{prop:degree-split}, we have that 
$\left|x(v)-x^{L}(v)\right| \leq 4\eps$ for each vertex $v\in V$.
Therefore, $d^{4\eps}_V(\x,\x^{L})=\sum_{v} \left(x(v)-x^{L}(v)-4\eps\right)^+ = 0$, and so $\x^L$ satisfies Property \ref{property:vertex-wise-slack} of a $(4\eps,\eps)$-coarsening of $\x$.

Finally, $\x^L$ is an $\epsilon'$-uniform fractional matching by definition, with $\eps'\in [\eps,2\eps]$, and thus satisfies Property \ref{property:edgebound-correct} of $(4\eps,\eps)$-coarsenings of $\x$. We conclude that $\calF_L$ is a $(4\eps,\eps)$-split of $\x$.
\subsection{Randomized Static Splitting}

\label{sec:api-rand-implement}

We now turn to proving \Cref{rand-split-whp}, restated below for ease of reference.

\randsplitwhp*

\begin{proof} 

Assume $\x$ is $\lambda$-uniform. First note that if $\lambda > \frac{\eps^4}{24 \cdot \log^2 n}$ then we may return $\{\x\}$ as the split trivially. Note that if $\lambda \leq \eps/n^2$ then $\norm{\x} \leq \eps$. In this case, it is sufficient to return a split consisting of $|\supp{\x}|$ coarsenings, where each coarsening is a single edge of $\supp{\x}$ and has weight $\eps^4/\log^2 n$. The same split can also be returned if $\norm{\x} \leq \eps^2/\log n$. Thus, from now on we assume that $\frac{\eps^4}{24 \cdot \log^2 n} \geq \lambda \geq \eps/n^2$ and $\norm{\bf{\x}} \geq \eps^2/\log n$. We start by defining the following two parameters.
$$\delta:=\frac{\eps^4}{24\log^2 n}  \text{ and } k:=2^{\lceil \log_2\frac{\delta}{\lambda}\rceil}.$$

\textbf{The algorithm:} We (implicitly) initialize $k$ zero vectors $\bf{\x}^{(1)},\dots,\bf{\x}^{(k)}$.
Next, for each edge $e\in \supp{\x}$, we roll a $k$-sided die $i \sim \Uni([k])$, and set $x^{(i)}_e \gets \delta':=\lambda \cdot k$.

\textbf{Running time:} The algorithm spends $O(1)$ time per edge in $\supp{x}$, and so it trivially takes $O(|\supp{x}|)$ time.

\textbf{Correctness:} By construction, 
$\{\suppOp(\x^{(i))}\}_{i=1}^k$ is a partition of $\supp{\x}$, and so each of these $\x^{(i)}$ satisfy Property \ref{property:containment} of coarsenings.
Moreover, since $\frac{\delta}{\lambda}\leq k\leq 2\cdot \frac{\delta}{\lambda}$, we have that $\delta'=k\cdot \lambda \in [\delta,2\delta]$. Hence, each of the vectors $\bf{\x}^{(i)}$ satisfy Property \ref{property:edgebound-correct} of an $(\eps,\delta)$-coarsening.
It remains to prove the two other properties of a coarsening.

For each edge $e\in \supp{\x}$ and $i\in [k]$, we have $x^{(i)}_e \sim \delta'\cdot \Ber(\frac{1}{k})$, 
and so 
$\E\left[x^{(i)}_e\right] = \frac{\delta'}{k} = \lambda$.
Thus, by linearity of expectation, we get: $\E[\normOp{\x^{(i)}}]=\norm{\x}]$.
Consequently, since $\norm{\x}\geq \frac{\eps^2}{\log n}$, by standard Chernoff bound and our choices of $\delta=\frac{\eps^4}{24\log^2 n}$ and $\delta'\leq 2\delta$,  
we have that Property \ref{property:global-slack}
holds w.h.p.
\begin{align*}
\Pr\left[|\;\normOp{\x^{(i)}} - \norm{\x} \;| \geq \eps\cdot \norm{\x} \right] & \leq 2\exp\left(-\frac{\eps^2 \cdot  \norm{\x}}{3\delta'}\right) \leq 2 n^{-4}.
\end{align*}
Similarly, $\E\left[x^{(i)}(v) \right] = x(v)$ for each vertex $v$, and so again a Chernoff bound implies that:
\begin{align*}
\Pr[|x^{(i)}(v) - x(v)| \geq \eps ]& \leq 2\exp\left(-\frac{\eps^2 \cdot  1}{3\delta'}\right) \leq 2n^{-4}. 
\end{align*}
Therefore, by taking a union bound over the $(n+2)$ bad events, for each of the $k \leq \frac{2\delta}{\lambda}\leq \frac{2\eps^4}{24\log^2 n}\cdot \frac{n^2}{\eps} \leq 2n^2$ vectors (i.e., $O(n^3)$ bad events), we have that $d^{\eps}_V(\x,\x^{(i)})=0$.
Therefore Property  \ref{property:vertex-wise-slack} also holds w.h.p., and so by union bound $\{\x^{(i)}\}$ satisfy all the properties of an $(\eps,\delta)$-split w.h.p.
\end{proof}

\subsection{Output-Adaptive Dynamic Coarsening}

\label{sec:proof:lem:dynamic-coursening:expectation}

In this section we prove \Cref{lem:dynamic-coursening:expectation}, using the following kind of \emph{set sampling} data structure.

\begin{wrapper}
\begin{restatable}{Def}{setsampler}\label{def:ds_set_sampler}
A \underline{\emph{dynamic
set sampler}} is a data structure supporting the following operations:
\begin{itemize}
\item $\init(n,\bf{p}\in[0,1]^{n})$: initialize the data structure for
 $n$-size set $S$ and probability vector $\bf{p}$.
\item $\set(i\in[n],\alpha\in[0,1])$: set $p_{i}\gets\alpha$.
\item $\sample()$: return $T\subseteq\R^{n}$ containing each $i\in S$ 
independently with probability $p_{i}$.
\end{itemize}
\end{restatable}
\end{wrapper}

\Cref{thm:set_sampler} shows that there exists a dynamic \outputadaptive set sampler with optimal properties. 
Concurrently to our work, another algorithm with similar guarantees was given by \cite{yi2023optimal}. 
As that work did not address the adaptivity of their algorithm, we present our (somewhat simpler) algorithm and its analysis, in \Cref{sec:app:setsampling}.

\begin{restatable}{thm}{thmsetsampler}\label{thm:set_sampler}
\Cref{alg:set_sampler} is a set
sampler data structure using $O(n)$ space that implements $\init$ in $O(n)$ time, $\set$ in $O(1)$ time,
and $T=\sample()$ in expected $O(1+|T|)$ time in a word RAM model with word size $w=\Omega(\log(p^{-1}_{\min}))$, under the promise that $p_i\geq p_{\min}$ for all $i\in [n]$ throughout. These
guarantees hold even if the input is chosen by an \outputadaptive adversary.
\end{restatable}

Equipped with \Cref{thm:set_sampler}, we are ready to prove \Cref{lem:dynamic-coursening:expectation}.

\randsplitexp*

\begin{proof}

\textbf{Initialization:} The algorithm maintains a set sampler as in \Cref{def:ds_set_sampler}, using the algorithm of \Cref{thm:set_sampler} over ${n \choose 2}$ elements. Each element with non-zero probability corresponds to an edge $e \in \suppOp(\x_{\leq \eps^3})$, and receives a probability of $x_e \cdot \eps^{-3}$ within the sampler. The algorithm initializes counter $C \leftarrow 0$. Afterwards the algorithm draws a sample $E$ from its set sampler and defines its output $\x'$ to have support $E \cup \suppOp(\x_{>\eps^3})$ and take weight $\eps^3$ on edges of $E$ and $x(e)$ for edges $e \in \suppOp(\x_{>\eps^3})$. The algorithm repeats this process until $\x_{\leq \eps^3}'$ is an $(100 \cdot \eps,\eps^3)$-coarsening of $\x_{\leq \eps^3}$. Define $\x^{(0)}$ to stand for the state of $\x$ at the last initialization. 

\textbf{Handling updates:} If an update occurs to $\x_{>\eps^3}$ the algorithm simply updates $\x'$ accordingly. If an update occurs to edge $e \in \x_{\leq \eps^3}$ then the algorithm first updates $e$-s weight within the set-sampler appropriately. If $e$ was deleted from $\supp{\x}$ (or equivalently $x_e$ was set to $0$) it is deleted from $\supp{\x'}$. If $e$ was inserted into $\supp{\x}$ (or equivalently $x_e$ was set to be some positive value from $0$) it is ignored. The algorithm increases the counter $C$ to $C + \eps^3$. If $C$ reaches $\norm{\x^{(0)}} \cdot \eps$ the algorithm re-initializes $\x'$ (through repeatedly sampling $E$ from its sampler until $\x_{\leq \eps^3}'$ is a $(100\eps, \eps^3)$-coarsening of $\x_{\leq \eps^3}$) and resets $C$ to $0$. 

\textbf{Correctness:} 
First, we note that as $\x' = \x_{\leq \eps^3}' + \x_{> \eps^3}$, then if $\x_{\leq \eps^3}'$ is an $(O(\eps), \eps^3)$-coarsening of $\x_{\leq \eps^3}$, then $\x'$ is an $(O(\eps), \eps^3)$-coarsening of $\x$.
Thus, we only need to argue that $\x_{\leq \eps^3}'$ remains an $(O(\eps), \eps^3)$-coarsening of $\x_{\leq \eps^3}$ throughout the run of the algorithm. This follows from Lemma~\ref{lem:coarsening:stability} as between two re-initializations at most $\norm{\x^{(0)}} \cdot \eps  \cdot \eps^{-3}$ updates occur.

\textbf{Update Time:}
Observe that apart from the cost of re-initialization steps, the algorithm takes $O(1)$ time to update its output. During a re-initialization the algorithm has to repeatedly draw edge sets from its sampler. Note that $\x_{> \eps^3}'$ remains unaffected during this process. Whenever the algorithm draws an edge sample $|E|$ by definition we have that $\E[|E| \cdot \eps^3] = \norm{\x_{\leq \eps^3}}$. We will later show that the algorithm only needs to draw samples $O(1)$ times in expectation. As re-initialization occurs after $\Omega(\norm{\x_{\leq \eps^3}} \cdot \eps^{-2})$ updates and each sample edge set is drawn in linear time with respect to its size we receive that the algorithm has $O(\eps^{-1})$ expected amortized update time. It remains to show that at each re-initialization the algorithm makes $O(1)$ calls to its sampler in expectation. We will show this by arguing that every time a sample is drawn $\x_{\leq \eps^3}'$ is an $(100\eps, \eps^3)$-coarsening of $\x_{\leq \eps^3}$ with constant probability. Note that the algorithm can triviarly check if this is indeed the case in $O(|\suppOp(\x_{\leq \eps^3})|)$ time.

First note that Property~\ref{property:containment} and Property~\ref{property:edgebound-correct} (of coarsenings) follow trivially from the definition of $\x'$. We will first argue that Property~\ref{property:global-slack} holds with constant probability at each attempt. This will aslo imply that the algorithm has $\init$ time $O(|\supp{\x}|)$ in expectation.

For any $e \in \supp{\x}$ let $X_e$ stand for the indicator variable of event that $e \in E$ (where $E$ stands for the random edge sample drawn by the algorithm), and let $\bar{X} = \sum X_e$. Note that $X_e$ are independently distributed random variables and $\E[\bar{X}] = \norm{\x_{\leq \eps^3}} \cdot \eps^{-3}$ by definition. Furthermore, $\norm{\x_{\leq \eps^3}'} = \bar{X} \cdot \eps^3$. Using standard Chernoff bounds we receive the following:
\begin{align*} \Pr\left[|\norm{\x_{\leq \eps^3}} - \norm{\x_{\leq \eps^3}'}| \leq 100 \cdot \eps \cdot \norm{\x_{\leq \eps^3}}  \right] & = 1 - \Pr\left[|\bar{X} - \E[\bar{X}]| \geq \E[\bar{X}] \cdot 100 \cdot \eps \right] \nonumber \\
     & \leq 1 - O\left(\exp\left(-\frac{(100 \cdot \eps)^2 \cdot \E[\bar{X}]}{3}\right)\right). \nonumber 
\end{align*}
Given $\E[\bar{X}] \geq \eps^{-2}/8$ (or equivalently $\norm{\x_{\leq \eps^3}} \geq \eps/8$), the probability in the last inequality is $\Omega(1)$. If, conversely, $\E[\bar{X}] \leq \eps^{-2}/8$ (or equivalently $\norm{\x_{\leq \eps^3}} \leq \eps/8$), then by a simple Markov's inequality argument we get the following: 
$$\Pr[|\norm{\x_{\leq \eps^3}} - \norm{\x_{\leq \eps^3}'}| \geq \eps] \leq \Pr\left[\norm{\x_{\leq \eps^3}'} \geq \frac{7\eps}{8}\right] \leq \frac{\E[\norm{\x_{\leq \eps^3}'}] \cdot 7}{8\eps} \leq \frac{1}{2}.$$
Either way, with constant probability $|\norm{\x_{\leq \eps^3}} - \norm{\x_{\leq \eps^3}'}| \leq \norm{\x_{\leq \eps^3}} \cdot \eps + \eps$, implying Property~\ref{property:global-slack} holds with constant probability. 

It remains to prove Property~\ref{property:vertex-wise-slack}. Let random variable $Y_v$ stand for $d_{\{v\}}^{100\eps}(\x_{\leq \eps^3}, \x_{\leq \eps^3}')$ and $\bar{Y}$ stand for $d_V^{100\eps}(\x_{\leq \eps^3}, \x_{\leq \eps^3}') = \sum Y_v$. We will first show that $\E[Y_v] \leq 10 \eps \cdot \x_{\leq \eps^3}(v)$. Summing over all vertices this yields that $\E[\bar{Y}] \leq \norm{\x_{\leq \eps^3}} \cdot 20 \cdot \eps$. By a simple Markov's role based argument we can conclude from there that $\Pr[d_V^{100\eps}(\x_{\leq \eps^3}, \x_{\leq \eps^3}') \leq 100 \cdot \eps \cdot \norm{\x_{\leq \eps^3}}] = \Omega(1)$.

Consider any edge $e \in E$ incident on $v$. Let $y_e$ be the indicator variable of the event that $e$ was sampled by the dynamic sampler on query. Hence, we have $\E[y_e] = x_e \cdot \eps^{-3} \leq 1$, and $y_e$ are independently distributed random variables. Let $y_v = \sum_{e :  v \in e} y_e$. It follows that $\E[y_v] = x_{\leq \eps^3}(v)/\eps^3$ and $y_v = x_{\leq \eps^3}'(v)/\eps^3$. First, assume that  $\eps^{i} \leq x_{\leq \eps^3}(v) \leq \eps^{i-1}$ for some $i \geq 3$. Then applying Chernoff bounds yields the following:
\begin{align*}
\Pr\left[|x_{\leq \eps^3}(v) - x_{\leq \eps^3}'(v)| \geq \eps\cdot k\right] &\leq \Pr\left[x_{\leq \eps^3}'(v) - x_{\leq \eps^3}(v) \geq \eps \cdot k\right] \nonumber \\
&\leq  \Pr\left[y_v \geq \E[y_v] \cdot (1 + (\eps^{2-i}-1) \cdot (k-1/2))\right] \nonumber \\
&\leq  \exp\left(-\frac{((\eps^{2-i}-1) \cdot (k-1/2))^{2} \cdot \E[y_v]}{1 + \eps^{2-i}}\right) \nonumber \\
&\leq  \frac{10 \cdot \eps}{k^3}. \nonumber
\end{align*}

Next, if $\eps^2 \leq x_{\leq \eps^3}(v)$ and $k \leq 1/\eps$, then a similar application of Chernoff bound gives:
\begin{align*}
\Pr\left[|x_{\leq \eps^3}(v) - x_{\leq \eps^3}'(v)| \geq \eps\cdot k\right] &\leq  \Pr\left[|E[y(v)] - y(v)| \geq \E[y(v)] \cdot \frac{k \cdot \eps^4}{x_{\leq \eps^3}(v)}\right] \nonumber \\
&\leq  \Pr\left[|\E[y(v)] - y(v)| \geq \E[y(v)] \cdot k \cdot \eps^2 \right] \nonumber \\
&\leq  2 \cdot \exp\left(- \frac{(k \cdot \eps^2)^2 \cdot \E[y(v)]}{3}\right) \nonumber \\
&\leq  2 \cdot \exp\left(- \frac{k^2 \cdot \eps^3}{3}\right) \leq  \frac{10}{k^3}. \nonumber
\end{align*}

Summing over all the possible values of $k$, we get:
\begin{align*}\E[d^{100\eps}_{V}(\x_{\leq \eps^3},\x_{\leq \eps^3}')] &=\E\left[(|x_{\leq \eps^3}(v) - x_{\leq \eps^3}'(v)| - \eps)^+\right] \\
& \leq \sum_{i = 1}^{1/\eps}  \Pr\left[|x_{\leq \eps^3}(v) - x_{\leq \eps^3}'(v)| \geq \eps k\right]  \cdot \eps  k  \\
& \leq \sum_{i = 1}^{1/\eps} \frac{10  \eps^2}{k^2} \\
& \leq 20\eps. \qedhere
\end{align*}
\end{proof}

\section{General Graphs}\label{sec:general}

In this section, we extend our algorithms for rounding dynamic fractional matchings from bipartite to general graphs. Our main result in this section is a formalization of the informal \Cref{thm:general-informal}.
Formalizing this theorem and its key subroutine requires some build up, which we present in \Cref{sec:general:prelim}.
For now, we restate the main application of this theorem, given by \Cref{thm:AMMresults}.

\AMMresults*

The rest of this section is organized as follows. In \Cref{sec:general:prelim}, we introduce some known tools from the dynamic matching literature, as well as one new lemma (\Cref{lem:AMFM}) which motivates us to coarsen known fractional matching algorithms.
In \Cref{subsec:extensions:coarsening} we provide our general-graph coarsening algorithms for structured fractional matchings.
Finally, in the last two subsections, we provide applications of these dynamic coarsening algorithms: computing AMMs, in \Cref{subsec:proof:theorem:general}, and rounding $\eps$-restricted fractional matchings (see \Cref{def:restricted}), in  \Cref{sec:restricted}.

\subsection{Section Preliminaries}
\label{sec:general:prelim}

At a high level, our approach for rounding fractional matchings in general graphs is very close to our approach for the same task in their bipartite counterparts.
However, as mentioned in \Cref{sec:intro}, since the fractional matching relaxation studied in prior works is not integral in general graphs,
we cannot hope to round arbitrary fractional matchings in general graphs.
Therefore, in general graphs we focus our attention on a particular structured family of fractional matchings, introduced by \cite{arar2018dynamic} and maintained dynamically by \cite{bhattacharya2017fully,bhattacharya2019deterministically}.
\begin{Def}\label{def:AMFM}
		A fractional matching $\x\in \mathbb{R}^E$ is \emph{$(\eps,\delta)$-almost-maximal ($(\eps,\delta)$-AMFM)} if for each edge $e\in E$, either $x_e \geq \delta$ or some $v\in e$ satisfies both $x(v) \geq 1-\eps$ and $\max_{f\in v} x_f \leq \delta$.
		%Finally, we call an $(\eps,1/\eps)$-AMFM  an \emph{$\eps$-AMFM}.
	\end{Def}

For the sake of convenience, and similarly to \Cref{obs:xmin-and-high-bits}, we will assume that $\x$ has minimum edge weight of $\x_{\min}\geq \eps\delta/n$. Observe that if all edges with weight below $\eps\delta/n$ are decreased to zero, then vertex weights change by at most $\eps$, resulting in a $(2\eps,\delta)$-AMFM, which for our needs is as good as an $(\eps,\delta)$-AMFM, as we soon illustrate. For similar reasons, we will assume throughout that $\eps$ is a power of two.
    
	By LP duality, an $(\eps,\delta)$-AMFM is a $(1/2-O(\eps))$-approximate fractional matching \cite{bhattacharya2015deterministic}.
	Moreover, prior work established that the support of an $(\eps,\delta)$-AMFM contains a $(1/2-O(\eps))$-approximate \emph{integral} matching (even in general graphs), provided $\delta=O(\frac{\eps^2}{\log n})$ \cite{arar2018dynamic,wajc2020rounding}.
    Specifically, they show that the support of such a fractional matching contains the following kind of integral matching sparsifier, introduced in \cite{bhattacharya2015deterministic}.
    %, analyzed by \cite{arar2018dynamic}.
	\begin{prop}\label{lem:kernel}
		An
        $(\eps,d)$-\emph{kernel}
        $K=(V,E_K)$ is a subgraph of $G=(V,E)$ satisfying:
		\begin{enumerate}
			\item $d_K(v)\leq d$ for every vertex $v\in V$.
			\item $\max_{v\in e} d_K(v)\geq d\cdot (1-\eps)$ for every edge $e\in E\setminus E_K$.
		\end{enumerate}
%		Such a subgraph $K$, termed an $(\eps,d)$-\emph{kernel} by \cite{bhattacharya2015deterministic,arar2018dynamic}, satisfies $\mu(K)\geq \frac{1}{2+O(\eps)}\cdot \mu(G).$
	\end{prop}

    Any kernel of a graph $G$ contain an $(1/2-\eps)$-approximate matching with respect to $G$ for $d$ sufficiently large \cite{arar2018dynamic}.
    An alternative proof of this fact was given by \cite{bhattacharya2023dynamic}, who show how to efficiently compute an $\eps$-AMM (almost maximal matching, itself a $(1/2-\eps)$-approximate matching) of the host graph in a kernel. Specifically, they show the following.

	\begin{prop}
        \label{prop:KERNELtoAMM}
		Given an $(\eps, d)$-kernel of $G = (V,E)$ with $d\geq \eps^{-1}$, one can compute an $\eps$-AMM in
		$G$ in deterministic update time $O(d \cdot \mu(G)\cdot \eps^{-1}\cdot \log (\eps^{-1}))$.
	\end{prop}
\begin{proof}
		The proof is given in \cite[Lemma 5.5]{bhattacharya2023dynamic}. The only missing detail is externalizing the dependence on $\eps$ of that lemma, which is that of the best static linear-time $(1-\eps)$-approximate maximum weight matching algorithm, currently $O(m\cdot \eps^{-1}\cdot \log (\eps^{-1}))$ for $m$-edge graphs \cite{duan2014linear}.\footnote{In the conference version of that paper, \cite{bhattacharya2023dynamic} do not state the requirement $d\geq \eps^{-1}$, but this requirement is necessary (and sufficient), as implied by the statement of \cite[Lemma 5.4]{bhattacharya2023dynamic}.}
\end{proof}
        Being able to periodically compute $\eps$-AMMs lends itself to dynamically maintaining $\eps$-AMMs, due to these matchings' natural \emph{stability}, as in the following lemma of \cite{bhattacharya2023dynamic}
    \begin{prop}
    \label{obs:AMMstability}
    Let $\eps \in (0, 1/2)$. If $M$ is an $\eps$-AMM in $G$, then the non-deleted edges of $M$ during any sequence of
    at most $\eps \cdot \mu(G)$ updates constitute a $6\eps$-AMM in $G$.
\end{prop}

    Previous work \cite{arar2018dynamic,wajc2020rounding} show that a the support of any $(\eps,\delta)$-AMFM contains a kernel, for $\delta$ sufficiently small. We now prove the simple fact that the support of an $(\eps,\delta)$-AMFM $\x$ \emph{is itself} a kernel, provided $\x_{\min}
    \geq \delta$.
	
	\begin{lem}\textbf{(When AMFMs are kernels).}\label{lem:AMFM}
		Let $\x$ be an $(\eps,\delta)$-AMFM of $G=(V,E)$ satisfying ${\x_{\min}} \geq \delta$. Then, $K=(V, \support{x})$ is an $(\eps,\delta^{-1})$-kernel of $G$.
	\end{lem}
 
\begin{proof}
		The degree upper bound follows from the condition $\x_{\min}\geq \delta$ together with the fractional matching constraint, implying that for each vertex $v\in V$,
		\begin{align*}
			d_K(v) = \sum_{f\in v} \mathds{1}[f\in \support{x}] \leq \sum_{f\in v} x_f\cdot \delta^{-1} \leq \delta^{-1}.
		\end{align*}
		For the lower bound, fix an edge $e\in E\setminus \support{x}$, which thus satisfies $x_e = 0 < \delta$. Therefore, by the $(\eps,\delta)$-AMFM property, some vertex $v\in e$ has $x(v)\geq 1-\eps$ and $\max_{f\in v} x_f \leq \delta$. But since $\x_{\min} \geq \delta$, this implies that each edge $f\ni v$ in $\support{x}$ is assigned value exactly $x_f = \delta$. Therefore, any edge $e\in E\setminus \support{x}$ has an endpoint $v$ with high degree.
		\begin{align*}
			d_K(v) & =  \mathds{1}[f\in \support{x}] = \sum_{f\in v} x_f\cdot \delta^{-1} \geq \delta^{-1}\cdot (1-\eps). \qedhere
		\end{align*}
\end{proof}

We are now ready to state this section's main result (the formal version of Theorem~\ref{thm:general-informal}), proved in \Cref{subsec:proof:theorem:general}. The key idea behind this lemma is that coarsenings of AMFMs yield kernels of a slightly larger graph, obtained by adding $O(\eps)\cdot \mu(G)$  many dummy vertices, and therefore allow us to periodically compute an $O(\eps)$-AMM $M$ in this larger graph, whch is then an $O(\eps)$-AMM in the current graph.
  
% \color{blue}  
\begin{restatable}{thm}{generalformal}\label{thm:general-formal}
    Let $\epsilon\in (0,1)$. Let $\calF$ be a dynamic $(\eps,\eps/16)$-AMFM algorithm with $\update$ time $t_f$ and output recourse $u_f$. Let $\calC$ be a dynamic $(\eps, \eps/16)$-coarsening algorithm with $\update$ time $t_c$ for vectors $\x$ satisfying $x_e\geq \eps/16$ implies that $(x_e)_i=0$ for $i>k+4$, and whose output $\x'$ on $\x$ satisfies $x'_e\in \{0,\eps/16\}$ if $x_e<\eps/16$. Then there exists a dynamic $\eps$-AMM algorithm $\calA$ with $\update$ time $O(\eps^{-3}\cdot \log(\eps^{-1})+t_f + u_f\cdot t_c)$. Moreover, $\calA$ is deterministic/\adaptive/\outputadaptive if both $\calF$ and $\calC$ are.
    \end{restatable}

While the above stipulations about $\calC$ and $\x$ may seem restrictive, as we shall later see, these are satisfied by our coarsening algorithms, and by known fractional matchings in general graphs (up to minor modifications).

To obtain \Cref{thm:AMMresults}, we apply \Cref{thm:general-formal} to the $O(\eps^{-2})$ time and output recourse $(\eps,\eps/16)$-AMFM algorithm of \cite{bhattacharya2019deterministically} (see \cite[Lemma C.1]{wajc2020rounding}), and the coarsening algorithms stated in the following lemma, presented in the following section. 

\begin{lem}
\label{lem:coarsening:AMFM}
        For $\epsilon=2^{-k}$ and $k\geq 0$ an integer, there exist dynamic $(\eps,\eps/16)$-coarsening algorithms for inputs $\x$ such that $x_e\geq \eps/16$ implies that $(x_e)_i=0$ for all $i>k+4$ that are:
    \begin{enumerate}
        \item Deterministic, with $\update$ time $O(\eps^{-1}\cdot ß(\log n + \log^{2}(\eps^{-1}))) = \tilde{O}(\eps^{-1}\cdot \log n)$.
        \item \expandafter\capitalize\adaptive, with $\update$ time $O(\eps^{-1} \cdot ((\log\log n)^2 + \log^2(\eps^{-1})))$,~succeeding~w.h.p.
        \item \expandafter\capitalize\outputadaptive, with expected $\update$ time $O\left(\eps^{-1}\cdot \log^2(\eps^{-1})\right) = \tilde{O}(\eps^{-1})$. 
    \end{enumerate}
    Moreover, these algorithms' output $\x'$ for input vector $\x$ satisfies $x'_e\in \{0,\eps/16\}$ if $x_e < \eps/16$.
\end{lem}

\subsection{Proof of \Cref{lem:coarsening:AMFM}: Coarsening in General Graphs}

\label{subsec:extensions:coarsening}

    In this section we show how to leverage our rounding and coarsening algorithms of previous sections to     
    efficiently 
    coarsen AMFMs, giving \Cref{lem:coarsening:AMFM}.

    First, we show that (the internal state of) \Cref{alg:dynamic_rounding_hierarchy} yields a dynamic coarsening algorithm.

\begin{lem}

    \label{cor:thm:general:coarsen}
    For $\eps=2^{-k}$ and $k\geq 0$ an integer, \Cref{alg:dynamic_rounding_hierarchy}  maintains a $(2\eps, \eps)$-coarsening $\x'$ of input dynamic vector $\x\in \mathbb{R}^E_{\geq 0}$ satisfying $(x_e)_i=0$ for all $i > k$ and edges $e$ with $x_e\geq \epsilon$. Moreover, for every edge $e\in E$, if $x_e < \epsilon$, then $x'_e \in \{0, \epsilon\}$.
    \end{lem}

\begin{proof}

    Recall that \Cref{alg:dynamic_rounding_hierarchy} when run on vector $\x$ maintains vectors $\x^{(0)}, \x^{(1)} \dots$. 
    We claim that $\x^{(k)}$ is a $(2\eps,\eps)$-coarsening of $\bf{x}$ at all times during the run of the algorithm. 
    Properties \ref{property:containment}, \ref{property:global-slack}
    and \ref{property:vertex-wise-slack}
    of $(2\eps,\eps)$-coarsenings are shown explicitly by Lemmas \ref{lem:sequential-containment}, \ref{lem:sequential-size} (and \Cref{lem:sequential-size-upper-bound}) and \ref{sequential-general}, respectively. 
    Finally, 
    for every edge $e$ with $x_e < \eps$, we have that $x^{(k)}_e\in \{0,\eps\}$, while for each edge $e$ with $x_e \geq \eps$, we have that $x^{(k)}_e = x_e - \sum_{i=k+1}^L (x_e)_i\cdot 2^{-i} = x_e$ (where the second equality follows from the lemma's hypothesis), which shows Property~\ref{property:edgebound-correct} of $(2\eps,\eps)$-coarsenings.
\end{proof}

Unfortunately, \Cref{alg:dynamic_rounding_hierarchy} is too slow to yield speedups on the state-of-the-art via \Cref{thm:general-formal}.
In contrast, the coarsening algorithms of  \Cref{sec:partial-rounding-implementations} are also insufficient for that theorem's needs, as they only maintain  $(\poly(\eps^{-1}\cdot \log n),\poly(\eps^{-1}\cdot \log n))$-coarsenings, and so to be relevant for \Cref{thm:general-formal}, they must be run with a much smaller error parameter $\eps' = \poly(\eps/\log n)$, and again be too slow to yield any speedups.

However, as we show, similar to our results for rounding bipartite fractional matchings, using the output coarsenings of the algorithms of \Cref{sec:partial-rounding-implementations} as the input of \Cref{alg:dynamic_rounding_hierarchy} yields fast and sufficiently coarse dynamic coarsenings. For this, we need the following simple lemma.

   \begin{lem}\textbf{(Coarsening composition).}\label{claim:composition}
        Let $\eps_1,\eps_2\in [0,1]$ and $\delta_1\leq \delta_2$. 
        If $\x^{(1)}$ is an $(\eps_1,\delta_1)$-coarsening of $\x$ and $\x^{(2)}$ is an $(\eps_2,\delta_2)$-coarsening of $\x^1$, then $\x^{(2)}$ is an  $(\eps_1 + 2\eps_2,\delta_2)$-coarsening~of~$\x$.
    \end{lem}
 \begin{proof}
    Since $\x^{(2)}$ and $\x^{(1)}$ are coarsenings of $\x^{(1)}$ and $\x$, respectively, we have that $\suppOp(\x^{(2)}) \subseteq \suppOp(\x^{(1)}) \subseteq \supp{\x}$, implying Property~\ref{property:containment}.  To prove that $\x^{(2)}$ satisfies properties \ref{property:global-slack} and \ref{property:vertex-wise-slack} of an $(\eps_1+2\eps_2, \delta_2)$-coarsening of $\x$, we first note that by Property \ref{property:global-slack}, $\normOp{x^{(1)}}\leq (1+\eps_1)\cdot \norm{x}$.
    Consequently, Property \ref{property:global-slack} follows by the triangle inequality as follows.
    \begin{align*}
        \left|\norm{x}-\normOp{\x^{(2)}}\right| & \leq \left|\normOp{\x}-\normOp{\x^{(1)}}\right| + \left|\normOp{\x^{(1)}}-\normOp{\x^{(2)}}\right| \\
        & \leq \eps_1 \cdot \normOp{\x} + \eps_1 + \eps_2 \cdot \normOp{\x^{(1)}} + \eps_2 & \textrm{Property \ref{property:global-slack}} \\
        & \leq \eps_1 \cdot \normOp{\x} + \eps_1 + \eps_2 \cdot (1+\eps_1) \cdot \normOp{\x} + \eps_2 \\
        & \leq (\eps_1+2\eps_2) \cdot \norm{x} + \eps_1+2\eps_2. & \eps_1\leq 1
    \end{align*}
    Property \ref{property:vertex-wise-slack} similarly follows from \Cref{obs:monotone-distance}.
    \begin{align*}
        d_V^{\eps_1 + 2\eps_2}(\x,\x^{(2)}) 
        & \leq d_V^{\eps_1 + \eps_2}(\x,\x^{(2)}) & \textrm{\Cref{obs:monotone-distance}}\\
        & \leq d_V^{\eps_1}(\x,\x^{(1)}) + d_V^{\eps_2}(\x^{(1)},\x^{(2)}) & \textrm{\Cref{obs:monotone-distance}}\\
        & \leq \eps_1 \cdot \normOp{\x} + \eps_1 + \eps_2 \cdot \normOp{\x^{(1)}} + \eps_2, & \textrm{Property \ref{property:vertex-wise-slack}}
    \end{align*}
    where we already established that the latter term is at most $(\eps_1+2\eps_2)\cdot \norm{x} + (\eps+2\eps_2)$.

    Finally, by Property~\ref{property:edgebound-correct},
    if $x_e\geq \delta_2(\geq \delta_1)$, then $x^{(2)}_e=x^{(1)}_e=x_e$ and otherwise $x^{(1)}_e < \delta_2$ and hence $x^{(2)}_e\in [\delta,2\delta)$.  That is, $\x^{(2)}$ satisfies Property~\ref{property:edgebound-correct} of an $(\eps_1+2\eps_2,\delta_2)$-coarsening of $\x$.
\end{proof}

    Given \Cref{claim:composition}, an algorithm reminiscent of the rounding algorithm of  \Cref{lem:dynamic-coarsening-to-rounding} allows us to combine two dynamic coarsening algorithms, as in the following dynamic composition lemma:

\begin{restatable}{lem}{coarsentoAMFM}
\label{cor:combinedpartialrounding}
\textbf{(Composing dynamic coarsenings).} Let $\eps_1,\eps_2\in [0,1]$ and $\delta_1\leq \delta_2$. For $i=1,2$, let $\calC_i$ be a dynamic $(\eps_i,\delta_i)$-coarsening algorithms with $\update$ times $t^{\calC_i}_U := t^{\calC_i}_U(\eps_i, \delta_i, n)$ and $\init$ times $O(|\supp{x}|\cdot t^{\calC_i}_I)$, for $t^{\calC_i}_I := t^{\calC_i}_I(\eps_i,\delta_i,n)$. Then, there exists a dynamic $(O(\eps_1+\eps_2),\delta_2)$-coarsening algorithm with $\update$ time $O(t_U^{\calC_1} + t^{\calC_2}_U + \eps^{-1}\cdot t^{\calC_2}_I)$ and $\init$ time $O(|\supp{x}|\cdot (t^{\calC_1}_I+t^{\calC_2}_I))$, which is deterministic/\adaptive/\outputadaptive if both $\calC_1$ and $\calC_2$ are. Moreover, it suffices for $\calC_2$ to be a coarsening algorithm only for inputs that are a subset of the support of an output of $\calC_1$.
\end{restatable}

As the proof is near identical to that of \Cref{lem:dynamic-coarsening-to-rounding}, we only outine one salient difference.
\begin{proof}[Proof (Sketch)]
    The algorithm broadly follows the logic of  \Cref{lem:dynamic-coarsening-to-rounding}, with $\calC_1,\calC_2,\eps_1$ and $\delta_1$ playing the roles of $\calC$,$\calR,\eps$ and $\delta$ in that lemma.
    The outline of the algorithm and the resulting running time and determinism/adaptivity are as in \Cref{lem:dynamic-coarsening-to-rounding}.
    The only difference in correctness analysis is that by \Cref{lem:coarsening:stability}, taking the non-deleted edges in $\suppOp(x^{(1)})\cap \{e\mid x_e<\delta_1\}$ periodically (i.e., after every $\eps_1\cdot\norm{x}/\delta_1$ updates and using them as input to $\suppOp(x^{(1)})\cap \{e\mid x_e<\delta_1\}$) allows us to maintain an $(O(\eps_1),\delta_1)$-coarsening of $\x$ as the (dynamic) input for $\calC_2$ throughout. (Note that $\calC_2$ is a dynamic coarsening algorithm for the resultant vector dominated by $\x^{(1)}$, by the lemma's hypothesis.) That the dynamic output of $\calC_2$ is an $(O(\eps_1+\eps_2),\delta_2)$-coarsening of $\x$ then follows \Cref{claim:composition}.
\end{proof}

    \Cref{lem:coarsening:AMFM} then follows from
    \Cref{cor:combinedpartialrounding}
    in much the same way that \Cref{thm:bipartite} follows from \Cref{lem:dynamic-coarsening-to-rounding} in \Cref{sec:coarsening-applications}, by taking the coarsening algorithms of 
    \Cref{sec:partial-rounding} as $\calC_1$ and \Cref{alg:dynamic_rounding_hierarchy} as the second coarsening algorithm $\calC_2$ (here, relying on \Cref{cor:thm:general:coarsen}). We omit the details to avoid repetition.

\subsection{Proof of Theorem~\ref{thm:general-formal}: Maintaining AMMs}
\label{subsec:proof:theorem:general}

In this section we show how to maintain $O(\epsilon)$-AMMs, by dynamically coarsening AMFMs (motivated by \Cref{lem:AMFM}), and periodically using these coarsened AMFMs to compute such AMMs. 

Our first lemma provides a static algorithm substantiating the intuition that a coarsening of an AMFM allows to efficiently compute an AMM (under mild conditions, which we address below).

    \begin{lem}\label{lem:AMM-from-coarsened-AMFM}
        Let $0\leq \delta\leq \epsilon\leq 1$. If $\mu(G)\geq 2\delta^{-1}\epsilon^{-1}$, then there exists a deterministic algorithm that, given an $(\eps,\delta)$-AMFM $\x$ of $G$ and an $(\epsilon,\delta)$-coarsening $\x'$ of $\x$ satisfying $x'_e\in \{0,\delta\}$ if $x_e<\delta$, outputs an $O(\eps+\delta)$-AMM of $G$ in time $O(\mu(G)\cdot \epsilon^{-1}\cdot \delta^{-1}\cdot \log(\eps^{-1}))$.
    \end{lem}
    \begin{proof}
		We modify $\x'$ within the desired time bounds to obtain an $(\epsilon'',\delta'')$-AMFM $\x''$ with $\x''_{\min}\geq \delta''$ (for $\eps''=\Theta(\epsilon+\delta)$ and $\delta''=\Theta(\delta)$ to be determined) in a graph $G''=(V\cup V'',E\cup E'')$ obtained from $G$ by adding some $|V''| = O(\epsilon'')\cdot \mu(G)$ many dummy vertices $V''$ as well as some dummy edges $E''$ to $G$.
		Therefore, by \Cref{lem:AMFM}, the support of $\x''$ is an $(\epsilon'',(\delta'')^{-1})$-kernel of $G''$. Hence, by \Cref{prop:KERNELtoAMM} we can compute deterministically an $\eps''$-AMM $M''$ of $G''$ in the claimed running time. 
		Now, by definition of an $\epsilon''$-AMM, $M''$ is maximal in a subgraph obtained from $G''$ after removing a set of vertices $U\subseteq V\cup V''$ of cardinality at most $|U|\leq \epsilon''\cdot \mu(G'')\leq \epsilon''\cdot (1+O(\eps''))\cdot \mu(G)\leq O(\epsilon'')\cdot \mu(G)$.
		Consequently, $M:=M''\cap E$ is maximal in a subgraph obtained from $G$ after removing at most $|U|+|V''|\leq O(\epsilon'')\cdot \mu(G)=O(\epsilon+\delta)\cdot \mu(G)$ many vertices, including the at most $O(\eps'')\cdot \mu(G)$ many vertices of $V$ matched by $M''$ to dummy vertices in $V''$.
		It remains to construct $G''$ and $\x''$.
		
		First, invoking \Cref{cl:bounded-coarsening}, we deterministically compute a bounded $(\epsilon',\delta)$-coarsening $\y$ of $\x$, for $\epsilon':=3(\epsilon+\delta)$, in time $O(|\suppOp(\x')|) \leq O(\delta^{-1}\cdot \norm{x}) \leq O(\delta^{-1}\cdot \mu(G))$. Here, the first inequality follows from $\x'_{\min}\geq \delta$, while the second follows from the integrality gap of the degree-bounded fractional matching polytope in general graphs implying that $\norm{x}\leq \frac{3}{2}\cdot \mu(G)$.
		By Property \ref{property:boundedcoursening} of bounded coarsenings, $y(v)\leq x(v)+\epsilon'\leq 1+\epsilon'$.
		We therefore consider $\z:=\y/(1+\epsilon')$, which by the above is a fractional matching with $\z_{\min}\geq \delta':=\delta/(1+\epsilon')$. Moreover, since $\norm{z} = \norm{y}/(1+\eps')\geq \norm{y}\cdot (1-\epsilon')$ and  hence $\sum_v |y(v)-z(v)|\leq 2\epsilon'\norm{y}$, we have that 
		\begin{align*}
			d^{15\eps'}_V(\x,\z) & \leq d^{\eps'}_V(\x,\y) + d^{14\eps'}_V(\y,\z) & \textrm{\Cref{obs:monotone-distance}}\\
			& \leq  d^{\eps'}_V(\x,\y) +d^0_{V}(\y,\z) & \textrm{\Cref{obs:monotone-distance}} \\
			& \leq   d^{\eps'}_V(\x,\y) + 2\epsilon'\norm{y} & \\
			& \leq \eps'\cdot \norm{\x}+\epsilon'+ 2\epsilon'\cdot ((1+\epsilon')\cdot \norm{x} + \epsilon') +2\epsilon'  & \textrm{Properties \ref{property:global-slack} and \ref{property:vertex-wise-slack}}\\
			& \leq  15\epsilon'\cdot \norm{x}+15\epsilon'. & \epsilon'=3(\epsilon+\delta)\leq 6.
		\end{align*}
		That is, $\z$ is a fractional matching that satisfies Property \ref{property:vertex-wise-slack} of a $(15\epsilon',\delta')$-coarsening of $\x$.
		
		Now, we say a node $v$ is \emph{useful} if $x(v)\geq 1-\epsilon$ and $\max_{f\ni v}x_f\leq \delta$ (noting their relevance in \Cref{def:AMFM}).
		By Property \ref{property:vertex-wise-slack}, and again using the integrality gap, we have that 
		\begin{align}\label{eqn:num-useful}
			\sum_{v \textrm{  useful}}(|x(v)-z(v)|-15\epsilon')^+\leq  d^{15\epsilon'}_V(x,z)\leq (15\epsilon'\cdot \norm{x}+15\epsilon') \leq 30\epsilon'\cdot \norm{x} \leq 45 \epsilon'\cdot \mu(G).
		\end{align} 

        We now create a set $V''$ of $90\cdot \epsilon'\cdot \mu(G)$ many dummy vertices, and 
		create a fractional matching $\x''$ from $\z$, letting $x''_e=z_e$ for each edge $e\in E$ and adding dummy edges $e\in E''$ with $\x''_e = \delta'':=\delta'= \delta/(1+\eps)\geq \delta/2$. Specifically, for each useful vertex $v$, we add some $\lceil (x(v)-z(v)-15\eps')^+\cdot (\delta'')^{-1}\rceil \leq 2x(v)\cdot (\delta'')^{-1}\leq 4\delta^{-1}$ dummy edges to distinct dummy vertices (this is where we use that $\eps\cdot \mu(G')\geq \delta^{-1}$), chosen in round-robin fashion. 
        This guarantees that each useful vertex $v$ has $x''(v)\geq x(v)-15\eps'$.
        That $\x''$ is a fractional matching, and in particular dummy nodes $u$ have $x(u)\leq 1$, follows from \Cref{eqn:num-useful}, as the amount of total $\x''$-value on dummy edges is at most $90\eps'\cdot \mu(G)\leq |V''|$, distributed in round-robin fashion.
		
		Finally, we note that any edge $e$ with $x''_e<\delta'$ has $z_e<\delta'$ and so $z_e=0$, by Property \ref{property:edgebound-correct}, and so $y_e=0$, implying in turn that $x_e<\delta$, again by Property \ref{property:edgebound-correct}, and so there exists a useful $v\in e$, by \Cref{def:AMFM}.
		Therefore, $x(v)\geq 1-\epsilon$, and so $x''(v)\geq x(v)-15\epsilon'\geq 1-\epsilon-15\epsilon'$.
		Moreover, all edges $f\ni v$ have $x_f\leq \delta$, and so $x'_f\in \{0,\delta\}$ by Property \ref{property:edgebound-correct} and the lemma's hypothesis. Consequently, $y_f\in \{0,\delta\}$, and so $x''_f\in \{0,\delta'\}$.
		That is, letting $\eps'':=\epsilon-15\epsilon'$ and $\delta'':=\delta'$, we conclude that $\x''$ is an $(\epsilon'',\delta'')$-AMFM with $\x''_{\min}\geq \delta''$ in the graph $G''$ obtained by adding at most $90\epsilon'\cdot \mu(G)=O(\epsilon+\delta)\cdot \mu(G)$ many dummy vertices to $G$. The lemma follows.
    \end{proof}

	The above lemma allows us to compute an AMM quickly when $\mu(G)$ is large. 
	The following simple complementary lemma allows us to compute a maximal matching quickly when $\mu(G)$ is small.
	
	\begin{lem}\label{fact:small-mu-MM}
		There exists a static deterministic  algorithm that given an $O(1)$-approximate vertex cover $U$ of graph $G=(V,E)$,  compute a maximal matching in $G$ in time $O(\mu(G)^2)$.
	\end{lem}
	\begin{proof}
		First, the algorithm computes a maximal matching $M$ in $G[U]$, taking $O(|U|^2)=O(\mu(G)^2)$ time, since an $O(1)$-approximate vertex cover has size $O(\mu(G))$.\footnote{The endpoints of any maximum matching form a vertex cover of size $2\mu(G)$.} Then, we extend $M$ to be a maximal matching in all of $G$ as follows: for each unmatched vertex $u\in U$, we scan the neighbors of $u$ until we find an unmatched neighbor $v$ (in which case we add $(u,v)$ to $M$) or until we run out of neighbors of $u$. As $|M|\leq \mu(G)$ at all times by definition of $\mu(G)$, we scan at most $2\mu(G)+1$ neighbors per node $u\in U$ until we match $u$ (or determine that all of its neighbors are unmatched), for a total running time of $O(U\cdot \mu(G))=O(\mu(G)^2)$. Since each edge in $G$ has an endpoint in the vertex cover $U$, this results in a maximal matching.
	\end{proof}

	We are finally ready to prove Theorem~\ref{thm:general-formal}, restated for ease of reference. 
    \generalformal*

	\begin{proof}
	Our algorithm dynamically maintains data structures, using which it periodically computes an $\epsilon$-AMM, the non-deleted edges of which serve as its matching during the subsequent period.
	A period at whose start the algorithm computes an $\epsilon$-AMM $M$ consists of $\lfloor\epsilon\cdot |M|\rfloor \leq \epsilon\cdot \mu(G)$ updates, and so by \Cref{obs:AMMstability}, the undeleted edges of $M$ are a $6\epsilon$-AMM throughout the period.
	We turn to describing the necessary data structures and analyzing their update time and the (amortized) update time of the periodic AMM computations.
    Throughout, we assume without loss of generality that $\eps=2^{-k}$ for $k\geq0$ some integer.
 
	\textbf{The data structures.} First, we maintain an $(\epsilon,\epsilon/16)$-AMFM $\x$ using update time $t_f$, and update recourse $u_f$. 
    For each edge $x_e\geq \eps/16$, we set $x_e=\eps/16$, noting that this does not affect the salient properties of an $(\eps,\eps/16)$-AMFM. This second assumption implies that $\calC$ is a coarsening algorithm for the dynamic AMFM we will wish to coarsen.
    This implies that coarsening algorithm $\calC$ can be run on $\x$ to maintain an $(\epsilon,\delta)$-coarsening $\x'$ of $\x$, where each change to $\x$ incurs update time $t_c$, for a total update time of $O(t_f+u_f\cdot t_c)$.
	In addition, using \cite{bhattacharya2019deterministically}, we deterministically maintain an $O(1)$-approximate vertex cover $U\subseteq V$ using $O(1)$ update time,\footnote{Generally, \cite{bhattacharya2019deterministically} maintain a $(2+\epsilon)$-approximate vertex cover in $O(\epsilon^{-2})$ for any $\epsilon\in (0,1)$.} which is dominated by the above update time.
	
	\textbf{The end of a period.} When a period ends (and a new one begins), both the current matching $M$ 
	and the new matching $M'$ that we compute for the next matching are $O(\epsilon)$-AMMs, and so $|M|=\Theta(\mu(G))$ and $|M'|=\Theta(\mu(G))$.
	We compute $M'$ as follows.
	If $|M|\leq \eps^{-2}$, we run the algorithm of \Cref{fact:small-mu-MM}, using the $O(1)$-vertex cover $U$, in time $O(\mu(G)^2)=O(|M'|^2)$, amortized over $\epsilon\cdot |M'|$ updates, for an amortized update time of $O(|M'|^2/(\epsilon\cdot |M'|))=O(\epsilon^{-1}\cdot |M'|)=O(\epsilon^{-3})$.
	Otherwise, we have that $(\mu(G)\geq)|M|>2\epsilon^{-2}$, and we use $\x$ and $\x'$ and \Cref{lem:AMM-from-coarsened-AMFM} to compute $M'$ in time $O(\mu(G)\cdot \epsilon^{-2}\cdot \log(\epsilon^{-1})) = O(|M'|\cdot \epsilon^{-2}\cdot \log(\epsilon^{-1}))$, for an amortized update time of $O(|M|\cdot \epsilon^{-2}\cdot \log(\epsilon^{-1}))/(\epsilon\cdot |M'|)=O(\epsilon^{-3}\cdot \log(\epsilon^{-1}))$.
	
	To summarize, the above algorithm maintains an $O(\epsilon)$-AMM within the claimed running time. As this algorithm consists of deterministic steps (including those of the vertex-cover algorithm of \cite{bhattacharya2019deterministically}) and runs algorithms $\calF$ and $\calC$, the combined algorithm is deterministic/\adaptive/\outputadaptive if the latter two algorithms are.
	\end{proof}

\subsection{Restricted Fractional Matchings}\label{sec:restricted}

To conclude this section, we briefly note that the same approach underlying \Cref{thm:general-formal} also allows us to round other known structured fractional matchings in general graphs. In particular, we can also round the following fractional matchings, introduced by \cite{bernstein2016faster}.
\begin{Def}\label{def:restricted}
Fractional matching $\x\in \mathbb{R}^E_{\geq 0}$ is \emph{$\eps$-restricted} if $x_e\in [0,\eps]\cup\{1\}$ for each~edge~$e\in E.$   
\end{Def}

The interest in restricted fractional matchings (in general graphs) is that their integrality gap is low, in the following sense (see, e.g., \cite{bernstein2016faster,assadi2022decremental}).
\begin{prop}\label{prop:restricted-matchings}
Let $\x$ be an
$\alpha$-approximate fractional matching in $G$ (i.e., $\norm{x}\geq \alpha\cdot \mu(G)$) that is also $\eps$-restricted. Then, $\supp{x}$ contains an $\alpha\cdot (1-\eps)$-approximate integral matching.
\end{prop}

\begin{restatable}{thm}{generalrestricted}\label{thm:general-restricted}
    Let $\eps=2^{-k}$ for $k\geq 11$ an integer.
    Let $\calF$ be a dynamic $\alpha$-approximate $\eps$-restricted fractional matching algorithm with update time $t_f$ and output recourse $u_f$. 
     Let $\calC$ be a dynamic $(16\eps, \eps)$-coarsening algorithm with $\update$ time $t_c$ for vectors $\x$ satisfying $x_e\geq \eps/16$ implies that $(x_e)_i=0$ for $i>k+4$. Then there exists a dynamic $\alpha(1-O(\eps))$-approximate matching algorithm $\calA$ with update time $O(\eps^{-3}+t_f + u_f\cdot t_r)$. Moreover, $\calA$ is deterministic/\adaptive/\outputadaptive if both $\calF$ and $\calC$ are.
\end{restatable}

The proof outline again mirrors that of \Cref{thm:general-formal} (and \Cref{lem:dynamic-coarsening-to-rounding} before it), maintaining essentially the same data structures, and so we only discuss the key difference.
\begin{proof}[Proof (Sketch)]
The maintained $(O(\eps),\eps)$-coarsening $\x'$ of the dynamic $\eps$-restricted fractional matching $\x$ is $2\eps$-restricted fractional matching, by Property \ref{property:edgebound-correct}. Moreover, by Property \ref{property:global-slack} and $\norm{x}\geq 1$, we have that $\normOp{\x'}\geq (1-32\eps)\norm{x}$.
Therefore, by \Cref{prop:restricted-matchings}, $\suppOp(x')$ contains an integral matching $M^*$ of cardinality $|M^*|\geq \normOp{\x'}\cdot (1-32\eps)\geq \norm{x}\cdot (1-32\eps)^2 \geq \alpha(1-1024\eps)\cdot \mu(G).$
Now, if instead of computing an AMM periodically, as in \Cref{thm:general-formal}, we compute an $\eps$-approximate matching $M$ in time $O(\eps^{-1}|\supp{x}|)=O(\eps^{-1}\norm{x}/\eps) = O(\eps^{-2}\cdot \mu(G))$.\footnote{For example, we can use the $\eps$-approximate maximum-weight matching algorithm of \cite{duan2014linear}, which runs in time $O(\eps^{-1}\cdot m)$ in $m$-edge unweighted graphs.} We  do so  every $\eps\cdot \mu(G)$ updates, thus maintaining a $(1-3\eps)$-approximation of $|M^*|$ throughout by ``stability'' of the matching problem (see \cite{gupta2013fully}. 
This then results in an $\alpha(1-1024\eps)\cdot(1-3\eps)\geq \alpha(1-O(\eps))$-approximation, with essentially the same update time of \Cref{thm:general-formal}, other than the $O(\eps^{-3}\cdot \log(\eps^{-1}))$ term in the update time, which is improved to $O(\eps^{-3})$.
\end{proof}
    
By combining \Cref{thm:general-restricted} with \Cref{lem:coarsening:AMFM}, we obtain efficient dynamic rounding algorithms for $\eps$-restricted fractional matchings, with update times and determinism/adaptivity as the  stated in \Cref{thm:AMMresults}

\section{Decremental Matching}\label{sec:decremental}

In this section we discuss applications of our rounding algorithms to speeding up decremental matching algorithms.
Prior results are obtained by a number of fractional algorithms \cite{bernstein2020deterministic,jambulapati2022regularized,bhattacharya2023dynamicLP,assadi2022decremental}, combined with known rounding algorithms (or variants thereof \cite{assadi2022decremental}).
We show how our (partial) rounding algorithms yield speeds ups on known decremental bipartite matching algorithms, and bring us within touching distance of similar deterministic results in general graphs.

\paragraph{Robust fractional matchings.} The approach underlying \cite{bernstein2020deterministic,jambulapati2022regularized,assadi2022decremental} is to phaseically compute a \emph{robust} fractional matching $\x$, in the sense that the value of $\x$ restricted to non-deleted edges remains $(1-\eps)$-approximate with respect to the current maximum (integral) matching in the subgraph induced by non-deleted edges, unless the latter decreases by a $(1-O(\eps))$ factor.
For example, in a complete graph, a uniform fractional matching is robust, while an integral matching is not, as deleting its edges would yield a $0$-approximation to the maximum matching, which would be unaffected by such deletions.
Formally, the above works implement the following.

\begin{Def}
    An \emph{\DM} decremental fractional matching algorithm 
    partitions the sequence of deletions into $P$ \emph{phases} of consecutive deletions (determined during the algorithm's run), and computes a fractional matching $\x^i$ at the start of phase $i$, where $\x^i$ restricted to non-deleted edges is a $(1-\eps)$-approximate throughout phase $i$ (i.e., until $\x^{i+1}$ is computed).
\end{Def}

As for dynamic algorithms, we may consider \DM decremental fractional matching algorithms that are either deterministic or adaptive. All currently known such algorithms are adaptive, and even deterministic.

In what follows, we show how our rounding algorigthms can be used to round known \DM decremental fractional matching algorithms to obtain new state-of-the-art decremental (integral) matching algorithms, starting with bipartite graphs.

\subsection{Bipartite Graphs}

Applying the rounding algorithms of \Cref{thm:bipartite}, one immediately obtains a framework to rounding such decremental fractional bipartite matching algorithms.

\begin{thm}\label{thm:decremental-framework}
    Let $\calF$ be an \DM decremental fractional bipartite matching algorithm, using total time $t^\calF$ and at most $p^\calF$ phases.
    Let $\calR$ be an $\eps$-approximate bipartite rounding algorithm with $\update$ time $t^\calR_U$ and $\init$ time $t^{\calR}_I(\x)$.
    Then, there exists a $(1-2\eps)$-approximate decremental bipartite matching algorithm $\calA$ which on graph starting with $m$ edges takes total time $O(t^\calF + \sum_{i=1}^{p^\calF} t^\calR_I(\x^i) + m\cdot t^\calR_U)$. Algorithm $\calA$ is deterministic/\adaptive/\outputadaptive if both $\calF$ and $\calR$ are.
 \end{thm}
\begin{proof}
    The combined algorithm $\calA$ is direct. Whenever $\calF$ computes a new fractional matching $\x^i$, we run $\init(G,\x^i,\eps)$ in algorithm $\calR$. Between computations of fractional matchings $\x^i$ and $\x^{i+1}$, we call $\update(e,0)$ for each edge $e$ deleted. The running time is trivially $O(t^\calF + \sum_{i=1}^{p^\calF} t^\calR_I(\x^i) + m\cdot t^\calR_U)$. As for the approximation ratio, the fractional matching $\x$ throughout a phase (i.e., the $\x^i$ value of undeleted edges between computation of $\x^i$ and $\x^{i+1}$) is $(1-\eps)$-approximate with respect to the current maximum matching $\mu(G)$, and the rounding algorithm $\calR$ maintains a matching $M\subseteq \supp{x}$ with the desired approximation ratio, 
    \begin{align*}
        |M| &\geq (1-\eps)\cdot \norm{x}\geq (1-\eps)^2\cdot \mu(G)\geq (1-2\eps)\cdot \mu(G).
    \end{align*}
    Finally, as Algorithm $\calA$ only uses algorithms $\calF$ and $\calR$, the former is deterministic/\adaptive/\outputadaptive if the latter two are.
\end{proof}

\Cref{thm:decremental-framework} in conjunction with \Cref{thm:bipartite} yield a number of improvements for the decremental bipartite matching problem (see \Cref{table:decremental}).\footnote{The rounding algorithm of \cite{bhattacharya2021deterministic} only states its $n$-dependence, which is $\log^4n$, but a careful study shows that their $\eps$-dependence is $\eps^{-6}$.} 
In particular, these theorems yield decremental (integral) bipartite matching algorithms with the same update time as the current best \emph{fractional} algorithms for the same problems \cite{bhattacharya2020deterministic,jambulapati2022regularized}.
Our improvements compared to prior work follows both from our faster update times (previously of the order $\Omega(\eps^{-4})$, and some with high polylog dependencies), as well as our rounding algorithms' initialization times, which avoid us paying the update time per each of the potentially $p^\calF\cdot m$ many edges in the supports of $\x^1,\x^2,\dots,x^{p^{\calF}}$.
We discuss the fractional algorithms of \cite{bernstein2020deterministic,jambulapati2022regularized}, substantiating the results in the subsequent table, in \Cref{app:decremental}.

\begin{center}	
	\begin{table}[h]
		{
		\small
		\begin{center}
			\centering
			
			\begin{tabular}{ | c | c | c | c | }
				\hline				
				New Time Per Edge & Fractional Algorithm & Closest Prior Best & Reference \bigstrut\\
                  \hline 
                {$\eps^{-2}\cdot n^{o(1)}$ (A)} & {\cite{jambulapati2022regularized} (see \Cref{lem:box})} & $\eps^{-4}\cdot n^{o(1)}$ (A)& \cite{jambulapati2022regularized} + \cite{wajc2020rounding} \bigstrut\\    
                \hline 
                 \multirow{ 2}{*}{$\eps^{-3}\cdot \log^{5} n$ (D)} & \multirow{ 2}{*}{\cite{jambulapati2022regularized} (see \Cref{lem:box})} & $\eps^{-4}\cdot \log^{5}n$ (A)& \cite{jambulapati2022regularized} + \cite{wajc2020rounding} \bigstrut\\  
                & & $\eps^{-9}\cdot \log^{9}n$ (D) & \cite{jambulapati2022regularized} + \cite{bhattacharya2021deterministic} \bigstrut\\  
                \hline 
                \multirow{ 2}{*}{$\eps^{-4}\cdot \log^3 n $ (D)} & \multirow{ 2}{*}{\cite{bernstein2020deterministic} (see \Cref{lem:congestion})} & $\eps^{-4}\cdot \log^3 n$ (A) & \cite{bernstein2020deterministic} + \cite{wajc2020rounding} \bigstrut\\  
                & & $\eps^{-10}\cdot \log^7 n$ (D) & \cite{bernstein2020deterministic} + \cite{bhattacharya2021deterministic} \bigstrut\\  
                \hline 
			\end{tabular}
		\end{center}}	
		\captionsetup{justification=centering}
		\caption{New results for $(1-\eps)$-approximate Decremental Bipartite Matching\\
        (D) and (A) stand for deterministic, and (randomized) adaptive, respectively}
		\label{table:decremental}
	\end{table}
\end{center}

\subsection{Potential Future Applications: General Graphs}

In \cite{assadi2022decremental}, Assadi, Bernstein and Dudeja follow the same robust fractional 
matching approach of \cite{bernstein2020deterministic}, and show how to implement it in \emph{general} graphs.
By assuming that $\mu(G)\geq \eps\cdot n$ throughout, which only incurs a (randomized) $\poly(\log n)$ slowdown, by known vertex sparsification results \cite{assadi2016maximum}, they show how to maintain an \DM fractional matching that is $\eps$-restricted (\Cref{def:restricted}) in total time $m\cdot 2^{O(\eps^{-1})}$. 
By adapting the rounding scheme of \cite{wajc2020rounding}, they then round this fractional matching randomly (but adaptively).

By appealing to \Cref{thm:general-restricted} and \Cref{lem:coarsening:AMFM}, we can similarly round $\eps$-restricted fractional matchings deterministically, with a $\log n\cdot \poly(\eps^{-1})$ update time. 
Combining with \emph{deterministic} vertex sparsification with a $n^{o(1)}$ slowdown \cite{kiss2022improving}, we obtain \underline{\emph{almost}} all ingredients necessary for the first sub-polynomial-time deterministic $(1-\eps)$-approximate decremental matching algorithm in general graphs.
Unfortunately, the fractional algorithm of \cite{assadi2022decremental} relies on randomization in one additional crucial step to compute an \DM fractional matching, namely in the subroutine $M$\textsc{-or-}$E^*()$, that outputs either a large fractional matching respecting some edge capacity constraints, or a minimum cut (whose capacities one then increased). For bipartite graphs, \cite{bernstein2020deterministic} gave a deterministic implementation of this subroutine, though for general graphs this seems more challenging, and so \cite{assadi2022decremental} resorted to randomization to implement this subroutine. Our work therefore leaves the de-randomization of this subroutine as a concrete last step to resolving in the affirmative the following conjecture.

\begin{conj}
    There exists a determinsitic decremental $(1-\eps)$-approximate matching algorithm with total update time $m\cdot n^{o(1)}$, for any constant $\eps>0$.
\end{conj}

\appendix
\section*{APPENDIX}
    \section{A $\dsplit$ Algorithm}\label{app:dsplit}

    In this section we provide an implementation and analysis of the $\dsplit$ algorithm guaranteed by \Cref{prop:degree-split}, restated below.
     \dsplitalgo*
\begin{proof}[Proof of \Cref{prop:degree-split}]
        First, suppose $G$ is a simple graph (i.e., it contains no parallel edges). We later show how to easily extend this to a multigraph with maximum multiplicity two.
        
        Recall that a \emph{walk} $\calW$ in a graph is a sequence of edges $(e_1,\dots,e_k)$, with $e_i=\{v_i,v_{i+1}\}$. We call $v_1$ and $v_{k+1}$ the \emph{extreme} vertices of $\calW$. We call odd-indexed and even-indexed edges $e_i$ \emph{odd} or \emph{even} for short.
        The walk $\calW$ is a \emph{cycle} if $v_1=v_{k+1}$, and the walk is \emph{maximal} if it cannot be extended, i.e., if there exists no edge $e=\{v_1,v\}$ or $e=\{v_{k+1},v\}$ outside of $\calW$.
        
        While $E\neq \emptyset$, Algorithm $\dsplit$ repeatedly computes maximal walks $\calW$, removes these walks' edges from $E$ and adds all odd edges (resp., even edges) of $\calW$ to the smaller (resp., larger) of $E_1$ and $E_2$, breaking ties in favor of $E_1$. This algorithm clearly runs in linear time, as the $O(|\calW|)$ time to compute a walk $\calW$ decreases $|E|$ by $|\calW|$.
        
        Next, since each edge is added to one of $E_1$ or $E_2$, we have that $|E|=|E_1|+|E_2|$. Moreover, since the number of odd and even edges differ by at most one for each walk (the former being no less plentiful), we have that $|E_1|\geq |E_2|$ and $\big||E_1|- |E_2|\big |\leq 1$ throughout. Property \ref{property:size-halved} follows.
        
        Finally, fix a vertex $v$ and $i\in \{1,2\}$.
        By maximality of the walks computed, once $v$ is an extreme vertex of at most one walk. Thus, all but at most two edges of $v$ are paired into successive odd/even edges in some walks computed. Consequently, 
        the number of odd and even edges of $v$ can differ by $0$, in which case  $d_i(v)=\frac{d_G(v)}{2}$, they may differ by $1$, in which case $d_i(v) \in \left\{\lfloor \frac{d_G(v)}{2}\rfloor,\ \lceil \frac{d_G(v)}{2}\rceil\right\}$, or they may differ by $2$, in which case $v$ has two more odd edges than even edges, the last walk $v$ belongs to is an odd-length cycle (and hence $G$ is not bipartite) and we have that $d_i(v)\in \left\{ \frac{d_G(v)}{2}-1,\,\frac{d_G(v)}{2}+1\right\}$.
        Properties \ref{property:degree-halved} and \ref{property:bipartite-degree-halved}  follow.

        Now, to address the (slightly) more general case where some edges have two parallel copies, we let $G'$ be the simple graph obtained after removing all parallel edges from $G$, and let $E'_1$ and $E'_2$ be the edge sets computed by running $\dsplit$ on $G'$. Then, for each pair of parallel edges $(e,e')$ in $G$, we arbitrarily add one of the pair to $E_1$ and add the other to $E_2$. Finally, we let $E_1\gets E_1\cup E'_1$ and $E_2\gets E_2\cup E'_2$. It is easy to see that the sets $E_1$ and $E_2$ are simple (by construction), and that these sets then satisfy all three desired properties with respect to $G$, since the edge sets $E'_1$ and $E'_2$ satisfy these properties with respect to $G'$, and every vertex $v$ has all of its (parallel) edges in $E\setminus E'$ evenly divided between $E_1$ and $E_2$. The linear running time is trivial given the linear time to compute $E'_1,E'_2$.
\end{proof}

\section{Dynamic Set Sampling}

\label{sec:app:setsampling}

In this section we provide an \outputadaptive data structure for the dynamic set sampling problem (restated below). Recall that this is the basic problem
of maintaining a dynamic subset of $[n]$ where every element is included
in the subset independently with probability $p_{i}$ under dynamic
changes to $p_{i}$ and re-sampling. This basic problem was studied by \cite{tsai2010heterogeneous,bringmann2017efficient,yi2023optimal}. 
\begin{wrapper}
\setsampler*
\end{wrapper}

Our main result of this section is that we can implement in each operation in total
time linear in the number of operations, $n$, and the size of the
$T$ output. 

\thmsetsampler*

\paragraph{Comparison with the concurrent work of \cite{yi2023optimal}.}
Our solution is somewhat simpler than that of the concurrent set sampler of \cite{yi2023optimal}. 
For our purposes, the only important difference is that our algorithm is provably \outputadaptive.\\

Our Algorithm~\ref{alg:set_sampler} and associated proof of Theorem~\ref{thm:set_sampler}
stem from a simple insight about computing when an element $i\in[n]$
will be in the output of $\sample()$. Note that if there are no operations
of the form $\set(i,\alpha)$, then the probability $i$ is in any
individual output of $\sample()$ is $p_{i}.$ Consequently, the probability
that $p_{i}$ is not in the output of $\sample()$ for the next $t$
calls to $\sample()$ is $(1-p_{i})^{t}$. Therefore, the number of
calls to $\sample()$ it takes for $i$ to be in the output sample
follows the geometric distribution with parameter $p_{i}$, i.e. $\Geo(p_{i})$!

Leveraging this simple insight above leads
to an efficient set sampler data structures. Naively, implementing set sampling
takes $O(n)$ time per call to $\sample()$, used to determine
for each element $i$ whether or not it should be in the output. 
However, we could instead simply sample from $\geom(p_{i})$ (in expected
$O(1)$, using \cite{bringmann2013exact}) whenever $\set(i,\alpha)$ is called or $i$ is in the output
of $\sample()$, in order to determine \emph{the next call} to $\sample()$ which will result
in $i$ being in the output. Provided we can efficiently keep track
of this information for each $i$, this would yield the desired bounds
in Theorem~\ref{thm:set_sampler}.

Unfortunately, when sampling from $\geom(p_{i})$, the output could
be arbitrarily large (albeit with small probability). Further, maintaining the data structure for
knowing when $i$ is scheduled to be in the output of $\sample()$
would naively involve maintaining a heap on arbitrary
large numbers, incurring logarithmic factors. There are many potential
data-structures and techniques to solve this problem. In our Algorithm~\ref{alg:set_sampler}
we provide one simple, straightforward solution. Every $n$ calls
to $\sample()$, we ``rebuild'' our data structure and rather than sampling
from $\geom(p_{i})$ to determine the next call to $\sample()$
that will output $i$, we instead simply sample to determine the next
call to $\sample()$ before the rebuild that will output $i$ (if
there is one). Algorithm~\ref{alg:set_sampler} simply does this,
resampling this time for $i$ whenever $\set(i,\alpha)$ is called.

\begin{algorithm}
\caption{Set Sampler Data Structure}
\label{alg:set_sampler}
\LinesNumbered
\SetKwInOut{State}{global}
\SetKwInOut{Return}{return}
\SetKwProg{Fn}{function}{}{}

\State{Size $n$ and $p\in[0,1]^{n}$}

\State{Current (relative) time $\tau$, subsets $T_{1},...,T_{n}\subseteq[n]$,
and next sample times $\tau_{1},...,\tau_{n}\in[n+1]$}

\BlankLine

\Fn{$\init(n,p\in[0,1]^{n})$}{

Save $n$ and $p$ as global variables\;
Initialize $\tau\gets1$, $T_{i}=\emptyset$, and $\tau_{i}=n+1$
for all $i\in[n]$\;

Call $\set(i,p_{i})$ for all $i\in S$\;

}

\BlankLine

\Fn{$\set(i\in[n],\alpha\in[0,1])$}{

\lIf{$\tau_{i}\neq n+1$}{$T_{\tau_{i}}\gets T_{\tau_{i}}\setminus\{i\}$}

\tcp{The loop below can be implemented in expected $O(1)$ time (see
Lemma~\ref{lem:set_sampler_time_efficient})}

\For{$j=\tau$ to $n$\label{line:set_sampler_for_start}}{

Independently with probability $p_{i}$, set $T_{j}\gets T_{j}\cup\{i\}$,
$\tau_{i}=j$, and \textbf{return} \label{line:set_sampler_for_end}

}

$\tau_{i}=n+1$\;

}

\BlankLine

\Fn{$\sample()$}{

Set $T\gets T_{\tau}$ and then set $\tau\gets\tau+1$\;

\lIf{$\tau<n+1$}{call $\set(i,p_{i})$ for all $i\in T$}

\lElse{set $\tau=1$, and then call $\set(i,p_{i})$ for all $i\in[n]$}

\Return{$T$}
}
\end{algorithm}

Algorithm~\ref{alg:set_sampler} is written without an efficient
determination of the next output time for each element, so that it is clear that this algorithm is a
set sampler. In the following Lemma~\ref{lem:set_sampler_time_efficient}
we show how to perform this efficient determination or, more precisely,
implementing the for-loop in Lines \ref{line:set_sampler_for_start}-\ref{line:set_sampler_for_end}
We then use this lemma to prove Theorem~\ref{thm:set_sampler}.
\begin{lem}
\label{lem:set_sampler_time_efficient} The for-loop in Lines \ref{line:set_sampler_for_start}-\ref{line:set_sampler_for_end}
of Algorithm~\ref{alg:set_sampler} can be implemented in expected
$O(1)$ time in the word RAM model.
\end{lem}

\begin{proof}
The loop executes the return statement with
a value of $j\in[t,n]$ with probability $p^{j-t}(1-p)$. Consequently,
if we let $\ell\geq0$ be sampled by the geometric distribution with
probability $p$, i.e., $\Pr[\ell=v]=p^{v}(1-p)$ for all $v\in\Z_{>0}$,
and if $\ell\in\{0,...,n-t\}$ simply execute the return statement
with $j=t+\ell$ and otherwise set $\tau_{i}=n+1$, then this is equivalent to the lines of the for loop. Since sampling from a
geometric distribution $\Geo(p)$ in this manner can be implemented in expected
$O(\log(1/p)/w)=O(1)$ time in the word RAM model \cite{bringmann2013exact}, the result follows.
\end{proof}

\begin{proof}[Proof of Theorem~\ref{thm:set_sampler}]
Algorithm~\ref{alg:set_sampler} maintains that after each operation ($\init$, $\set$ or $\sample$),
each $i\in[n]$ is a member of at most one $T_{j}$. Further, if $i\in T_{j}$,
then $\tau\leq j$ and $\tau_{i}=j$ and moreover $\tau_{i}=n+1$ if and
only if $i\notin T_{j}$ for any $j\in[n]$. 
Further, the algorithm is designed (as discussed) so 
that $T_{\tau}$ is a valid output of $\sample$ at time $\tau$  (for any updates of an \outputadaptive adversary, that is unaware of $T_j$ for $j\geq \tau$).
Since the algorithm also ensures that $\tau\leq n$, the algorithm
has the desired output. Further, from these properties it is clear
that the algorithm can be implemented in $O(n)$ space. It only remains
to bound the running time for implementing the algorithm.

To analyze the running time of the algorithm, first note that by Lemma~\ref{lem:set_sampler_time_efficient},
each $\set$ operation can be implemented in expected $O(1)$ time.
Consequently, $\init$ can be implemented in expected $O(n)$ time.
Further, since $T=\sample()$ simply calls $\set$ for elements in
its output or for all $n$ elements after every $n$ times it is called, it
has the desired expected runtime $O(1+|T|)$ as well.
\end{proof}

\section{Known Decremental Fractional Bipartite Matching Algorithms}\label{app:decremental}

For completeness, we briefly discuss the parameters of the decremental bipartite algortihms of \cite{bernstein2020deterministic,jambulapati2022regularized}, which together with \Cref{thm:decremental-framework} and \Cref{thm:bipartite} yield the results of \Cref{table:decremental}.

\begin{lem}[Congestion Balancing \cite{bernstein2020deterministic}]\label{lem:congestion}
    There exists an \DM decremental fractional bipartite matching algorithm with total time $O(m\cdot \eps^{-4}\cdot \log^3n)$ and $O(\eps^{-1}\cdot \log n)$ phases.
\end{lem}
\begin{proof}
    The algorithm of \cite{bernstein2020deterministic} consists of $O(\eps^{-1}\cdot \log n)$ phases, where a phase starts when the maximum matching size of $G$ decreases by a factor of $(1-\eps)$. The bottleneck of each phase is the $O(m\cdot \eps^{-2}\cdot \log^3 n)$-time subroutine \textsc{Robust-Matching}, which computes a fractional matching that does not change during the phase (aside from nullification of value on deleted edges), and is guaranteed to remain $(1-\eps)$-approximate during the phrase \cite[Lemma 5.1]{bhattacharya2020deterministic}.
    This subroutine's running time is $O(m\cdot \eps^{-3}\cdot \log^2n)$. Thus, the algorithm's total running time is $O(m\cdot \eps^{-4}\cdot \log^3n)$.
\end{proof}

\begin{lem}[Regularized Box-Simplex Games \cite{jambulapati2022regularized}]\label{lem:box}
    There exist \DM  decremental fractional bipartite matching algorithms that are:
    \begin{enumerate}
        \item Deterministic, with total time $O(m\cdot \eps^{-3}\cdot \log^{5} n)$ and $O(\epsilon^{-2} \cdot \log^2 n)$ phases \cite[Theorem 7]{jambulapati2022regularized}.
        \item \expandafter\capitalize\adaptive, with total time $O(m\cdot \eps^{-2}\cdot n^{o(1)})$ and $O(\epsilon^{-2} \cdot \log^2 n)$ phases \cite[Theorem 9]{jambulapati2022regularized}.
    \end{enumerate}
\end{lem}

\begingroup
\sloppy

\section*{Acknowledgements}

Thank you to Arun Jambulapati for helpful conversations which were foundational for the static bipartite rounding algorithm in this paper and subsequent developments. Part of this work was conducted while authors in positions $2^0, 2^1$ and $2^2$ (alphabetically) were visiting \href{https://www.dagstuhl.de/en/seminars/seminar-calendar/seminar-details/22461}{Dagstuhl program 22461}, \href{https://www.dagstuhl.de/en/seminars/seminar-calendar/seminar-details/22461}{``Dynamic Graph Algorithms''}.
\endgroup

\bibliographystyle{alpha}
\bibliography{abb,ultimate}

\end{document}